%% file: arXiv.tex
\documentclass[a4paper,USenglish, cleveref,nameinlink,autoref,thm-restate,authorcolumns]{lipics-v2021}

\newcommand{\removedforshortening}[1]{\xspace}
\newcommand{\rfs}[1]{\xspace}

\graphicspath{{./figures/}}

\newtheorem{reductionrule}{Rule}
\Crefname{reductionrule}{Rule}{Rules}

\newcommand{\PDE}{{\sc PDE}\xspace}
\newcommand{\PDEinF}{{\sc PDEinF}\xspace}
\newcommand{\PDEoutF}{{\sc PDEoutF}\xspace}
\newcommand{\threesat}{{\sc 3-SAT}\xspace}

\newcommand{\monthreesat}{{\sc Monotone Planar 3-SAT}\xspace}
\newcommand{\pt}{P$_2$-pendant\xspace}

\newcommand{\true}{\texttt{true}\xspace}
\newcommand{\false}{\texttt{false}\xspace}

\DeclareMathOperator{\orient}{orient}
\DeclareMathOperator{\interior}{interior}
\DeclareMathOperator{\nocross}{NoCross}
\DeclareMathOperator{\cyclic}{cyclic}

\usepackage{microtype}
\usepackage{gensymb}
\usepackage{slantsc}
\usepackage{lineno}

\usepackage{paralist}

\definecolor{defblue}{rgb}{0.121,0.47,0.705}
\DeclareTextFontCommand{\emph}{\color{defblue}\em}

\usepackage[most,many,breakable]{tcolorbox}
\usepackage{xcolor}
\usepackage{varwidth}
\definecolor{lipicsblue}{rgb}{0.08235294118,0.3098039216,0.537254902}
\definecolor{lipicsyellow}{HTML}{e6b919}
\definecolor{linkblue}{rgb}{0.098,0.098,0.4392}
\definecolor{ourgreen}{rgb}{0.509,0.745,0.235}
\definecolor{indianred}{rgb}{0.804,0.361,0.361}
\definecolor{indianred1}{rgb}{1,0.416,0.416}
\definecolor{indianred3}{rgb}{0.804,0.333,0.333}
\definecolor{orangered}{rgb}{1,0.271,0}
\definecolor{coral1}{rgb}{1,0.447,0.337}
\definecolor{rosybrown2}{rgb}{0.933,0.231,0.231}
\definecolor{aquamarine4}{rgb}{0.271,0.545,0.455}
\definecolor{chartreuse3}{rgb}{0.4,0.804,0}
\definecolor{mediumpurple3}{rgb}{0.537,0.408,0.804}
\definecolor{mediumvioletred}{rgb}{0.78,0.082, 0.522}
\definecolor{realred}{rgb}{1,0,0}

\hypersetup{colorlinks=true,
	linkcolor=aquamarine4,
	anchorcolor=lipicsblue,
	citecolor=lipicsblue,
	filecolor=lipicsblue,
	menucolor=lipicsblue,
	urlcolor=lipicsblue,
	bookmarksopen=true,
	bookmarksopenlevel=2,
	bookmarksnumbered=true,
	plainpages=false,
}

\newtcolorbox{Definition}[2]{enhanced,
	before skip=2mm,after skip=2mm, colback=lipicsyellow!5!white,colframe=lipicsyellow!100!white,boxrule=0.5mm,
	attach boxed title to top left={xshift=1cm,yshift*=1mm-\tcboxedtitleheight}, varwidth boxed title*=-3cm,
	boxed title style={frame code={
			\path[fill=tcbcolback]
			([yshift=-1mm,xshift=-1mm]frame.north west)
			arc[start angle=0,end angle=180,radius=1mm]
			([yshift=-1mm,xshift=1mm]frame.north east)
			arc[start angle=180,end angle=0,radius=1mm];
			\path[left color=tcbcolback!60!white,right color=tcbcolback!60!white,
			middle color=tcbcolback!80!white]
			([xshift=-2mm]frame.north west) -- ([xshift=2mm]frame.north east)
			[rounded corners=1mm]-- ([xshift=1mm,yshift=-1mm]frame.north east)
			-- (frame.south east) -- (frame.south west)
			-- ([xshift=-1mm,yshift=-1mm]frame.north west)
			[sharp corners]-- cycle;
		},interior engine=empty,
	},
	fonttitle=\bfseries,
	colbacktitle=lipicsyellow!100!white, title={#1}}

\Crefname{section}{Sect.}{Sects.}
\Crefname{figure}{Fig.}{Figs.}
\Crefname{observation}{Observation}{Observations}

\definecolor{myg}{RGB}{56, 140, 70}
\tcbuselibrary{theorems,skins,hooks}
\newtcbtheorem[number within=section]{Theorem}{DP}
{%
	enhanced,
	breakable,
	colback = myg!10,
	frame hidden,
	boxrule = 0sp,
	borderline west = {2pt}{0pt}{myg},
	sharp corners,
	detach title,
	before upper = \tcbtitle\par\smallskip,
	coltitle = myg!85!black,
	fonttitle = \bfseries\sffamily,
	description font = \mdseries,
	separator sign none,
	segmentation style={solid, myg},
	title={#1}
}
{th}

\tcbuselibrary{theorems,skins,hooks}
\newtcbtheorem[number within=section]{Complexity}{Complexity}
{%
	enhanced,
	breakable,
	colback = lipicsyellow!10,
	frame hidden,
	boxrule = 0sp,
	borderline west = {2pt}{0pt}{lipicsyellow},
	sharp corners,
	detach title,
	before upper = \tcbtitle\par\smallskip,
	coltitle = lipicsyellow!85!black,
	fonttitle = \bfseries\sffamily,
	description font = \mdseries,
	separator sign none,
	segmentation style={solid, lipicsyellow},
	title={#1}
}
{th}

\usepackage{graphicx} 
\usepackage{hyperref}
\usepackage[textsize=footnotesize]{todonotes}

\usepackage{amsmath}
\usepackage{caption}
\usepackage{xspace}
\usepackage{complexity}
\usepackage{braket}
\usepackage{amsmath,amsfonts,amssymb}

\usepackage{algorithm}
\usepackage{algpseudocode}
\usepackage[linewidth=1pt]{mdframed}
\usepackage{breqn}

\DeclareMathOperator{\vc}{vc}

\hideLIPIcs

\usepackage{url} 
\usepackage{hyperref}
\usepackage{xcolor,soul}

\input{fancyProblem}

\title{Extending Biconnected Straight-Line Planar Drawings}

\author{Giordano Andreola}{Roma Tre University, Rome, Italy
\and
\url{https://compunet.ing.uniroma3.it/\#!/people/andreola}
}{giordano.andreola@uniroma3.it}{https://orcid.org/0009-0003-7406-3514}{}

\author{Susanna Caroppo}{Roma Tre University, Rome, Italy
\and
\url{https://compunet.ing.uniroma3.it/\#!/people/caroppo}
}{susanna.caroppo@uniroma3.it}{https://orcid.org/0009-0001-4538-8198}{}

\author{Giordano {Da Lozzo}}{Roma Tre University, Rome, Italy
\and
\url{https://uniroma3.gitlab.io/compunet/gdl/}
}{giordano.dalozzo@uniroma3.it}{https://orcid.org/0000-0003-2396-5174}{}

\author{Marco {D'Elia}}{Roma Tre University, Rome, Italy
\and
\url{https://compunet.ing.uniroma3.it/\#!/people/delia}
}{marco.delia@uniroma3.it}{https://orcid.org/0009-0008-6266-3324}{}

\author{Giuseppe {Di Battista}}{Roma Tre University, Rome, Italy
\and
\url{https://compunet.ing.uniroma3.it/\#!/people/gdb}
}{giuseppe.dibattista@uniroma3.it}{https://orcid.org/0000-0003-4224-1550}{}

\author{Fabrizio Frati}{Roma Tre University, Italy
\and
\url{https://compunet.ing.uniroma3.it/\#!/people/frati}
}{fabrizio.frati@uniroma3.it}{https://orcid.org/0000-0001-5987-8713}{}

\author{Fabrizio Grosso}{CeDiPa - University of Perugia, Perugia, Italy
\and
\url{https://compunet.ing.uniroma3.it/\#!/people/grosso}
}{fabrizio.grosso@unipg.it}{https://orcid.org/0000-0002-5766-4567}{}

\author{Maurizio Patrignani}{Roma Tre University, Rome, Italy
\and
\url{https://compunet.ing.uniroma3.it/\#!/people/titto}
}{maurizio.patrignani@uniroma3.it}{https://orcid.org/0000-0001-9806-7411}{}

\authorrunning{Andreola, Caroppo, {Da Lozzo}, {D'Elia}, {Di Battista}, Frati, Grosso, Patrignani}

\Copyright{Giordano Andreola, Susanna Caroppo, Giordano {Da Lozzo}, Marco {D'Elia}, Giuseppe {Di Battista}, Fabrizio Frati, Fabrizio Grosso, and Maurizio Patrignani}



\funding{The seventh author was supported by Ce.Di.Pa. - PNC Programma unitario di interventi per le aree  del terremoto del 2009-2016 - Linea di intervento 1 sub-misura B4 - "Centri di ricerca per l'innovazione" CUP J37G22000140001.}

\nolinenumbers

\begin{document}

\maketitle

\begin{abstract}
The \textsc{Partial Drawing Extensibility} problem, for short \PDE, takes as input a triple $\langle G,H,\Gamma_H\rangle$, where $G$ is a planar graph, $H$ is a subgraph of $G$, and $\Gamma_H$ is a straight-line planar drawing of $H$, and asks whether $\Gamma_H$ can be extended to a straight-line planar drawing of $G$. Patrignani [Int. J. Found. Comput. Sci. (2006)] proved that the \PDE problem is \NP-hard, exploiting instances in which $H$ is highly disconnected. In this paper, we study the \PDE problem under the requirement that the initial partial drawing $\Gamma_H$ is biconnected. We show that \PDE remains \NP-hard even for instances in which $H$ is a biconnected graph with faces of bounded size, $G$ is subcubic, and the part of $G$ that is not in $H$ consists of length-$2$ paths.
The complexity of \PDE remains however open when $H$ is connected (or even biconnected) if $G$ has a fixed embedding. 
In this setting both a polynomial-time algorithm or an \NP-hardness proof seem to be elusive targets. 
As a step towards tackling this problem, we study instances of \PDE in which $H$ is biconnected, $G$ has a fixed embedding, and the rest of the graph consists of $p$ length-2 paths, and present an $O(p^2 n)$-time algorithm, a result in sharp contrast with the \NP-hardness of the variable embedding setting.
Moreover, with an approach based on the Existential Theory of the Reals, we show that, if $H$ is biconnected, the problem is FPT parameterized by the vertex cover number of $G$, both in a fixed and in a variable embedding setting.

\end{abstract}


\section{Introduction}\label{se:intro}
We study a problem at the intersection of two classical topics in Graph Drawing: straight-line planar drawings and drawing extensibility.
It is a fundamental result that every planar graph admits a straight-line planar drawing~\cite{fary1948straight,stein1951convex,wagner1936bemerkungen}.
The construction of such drawings has been extensively studied under various constraints; for example, requiring faces to be convex polygons~\cite{chiba1984linear,Tutte-convex-60,Tutte-how-draw-graph-63}, or requiring vertex coordinates to be integral while minimizing the drawing area~\cite{DBLP:journals/combinatorica/FraysseixPP90,DBLP:conf/soda/Schnyder90}.
Graph drawing extension problems involve extending a prescribed partial drawing, with specific properties, of a subgraph $H$ of a graph $G$ to a complete drawing of $G$ while maintaining those properties. Some extension problems are purely topological, like the one of extending a given planar embedding of $H$ to a planar embedding of $G$~\cite{DBLP:journals/talg/AngeliniBFJKPR15,DBLP:journals/comgeo/JelinekKR13}.
%
%
Many extension problems with a more geometric flavor have been studied, including the extension of:
%
%
visibility representations, where vertices are horizontal segments and edges are vertical segments~\cite{DBLP:conf/gd/ChaplickGGKL16,DBLP:journals/algorithmica/ChaplickGGKL18}; orthogonal drawings, where vertices are represented by points and edges are represented by orthogonal chains~\cite{DBLP:conf/gd/AngeliniRP20,DBLP:journals/jgaa/AngeliniRS21, DBLP:journals/jocg/BhoreGKMN23};
level and ordered level drawings, where vertices lie on horizontal lines and edges are monotone curves~\cite{DBLP:journals/tcs/BrucknerR25,DBLP:journals/talg/KlemzR19};
%
%
interval representations, where vertices are intervals on a straight line and adjacencies correspond to interval intersections~\cite{DBLP:journals/algorithmica/KlavikKORSSV17,DBLP:journals/algorithmica/KlavikKOSV17}; 
circle representations, where vertices are chords of a circle and edges correspond to intersections of such chords~\cite{BrucknerRS24,DBLP:conf/gd/ChaplickFK13,ChaplickFK19}; rectangular duals, which are contact representations by axis-aligned rectangles, no four of which share a point and whose union is a rectangle~\cite{DBLP:journals/tcs/ChaplickFKKRW22};
upward drawings of digraphs, where edges are monotone curves~\cite{DBLP:journals/comgeo/LozzoBF20}; 
1-planar drawings, where edges are allowed to cross at most once~\cite{DBLP:conf/mfcs/EibenGHKN20, DBLP:conf/icalp/EibenGHKN20}; simple drawings, where every two edges share at most one point (either a crossing point or an endpoint) and no three edges may cross in a single point~\cite{DBLP:conf/gd/ArroyoDP19}; and stack layouts, where vertices are points along a line and edges are semicircles drawn on half-planes delimited by the line~\cite{DBLP:conf/gd/DepianFGN24}. For further results regarding extensions problems refer, e.g., to~\cite{DBLP:journals/dcg/ArroyoKPVSW23,DBLP:conf/icalp/GanianHKPV21,DBLP:journals/algorithmica/GutwengerMW05,DBLP:conf/compgeom/HammH22, DBLP:journals/corr/abs-2412-13092, DBLP:conf/gd/KathederKKPR24}.

In this paper we focus on the \textsc{Partial Drawing Extensibility} problem, in which $G$ is a planar graph and a straight-line planar drawing $\Gamma_H$ of a subgraph $H$ is given; the goal is to extend $\Gamma_H$ to a straight-line planar drawing of $G$. The combinatorial core of this problem has its roots into Tutte's celebrated result~\cite{Tutte-convex-60} proving that every $3$-connected planar graph $G$ has a straight-line planar convex drawing for an arbitrary convex polygon~$\Gamma_H$ representing the cycle $H$ bounding its outer face. Several extensions of this result have been proposed over the years, see, e.g.,~\cite{DBLP:conf/gd/ChambersEGL10,DBLP:journals/jgaa/ChambersEGL12,DBLP:journals/dam/HongN08,DBLP:journals/jocg/OpheldersRSV25}. From a computational complexity perspective, Patrignani~\cite{DBLP:conf/gd/Patrignani05a,p-epsld-06} proved that the \textsc{Partial Drawing Extensibility} problem is \NP-hard. This hardness was subsequently strengthened~\cite{DBLP:conf/gd/LubiwMM18,DBLP:journals/jgaa/LubiwMM22}, by proving that the problem is complete for the Existential Theory of the Reals; for this result, $\Gamma_H$ defines a (not necessarily simple) polygonal region in which a straight-line planar drawing of $G$ has to be constructed (with edges possibly overlapping the boundary of the region). Conversely, the problem is linear-time solvable if $G$ has a fixed planar embedding, $H$ is a biconnected graph, and $\Gamma_H$ is a convex drawing~\cite{DBLP:journals/algorithmica/MchedlidzeNR16}.

\subparagraph{Our Contributions.} 
We study the \textsc{Partial Drawing Extensibility} (\PDE) problem for instances $\langle G,H,\Gamma_H\rangle$ in which $H$ is biconnected. The \NP-hardness proof of Patrignani~\cite{DBLP:conf/gd/Patrignani05a,p-epsld-06} employs instances in which $H$ is highly disconnected and $G$ has a variable embedding. Whereas it is not difficult to modify the proof in~\cite{DBLP:conf/gd/Patrignani05a,p-epsld-06} so that $G$ has a prescribed embedding, the proof also strongly relies on occlusions of visibility lines in $\Gamma_H$ by distinct components of $H$, and a modification of the construction to make $H$ connected appears to be elusive. 

In this paper, we give a new \NP-hardness proof for \PDE on instances such that $H$ is biconnected, its faces have bounded size, and $G$ is subcubic, thus providing a significant strengthening of the 20-year-old proof by Patrignani. The proof exploits the fact that~$G$ has a variable embedding, but holds in the restricted setting in which the ``undrawn'' part of~$G$ only consists of a set of length-$2$ paths. We also prove that, in this setting, if we drop the freedom of choice for the embedding of~$G$ the problem becomes polynomial-time solvable. Namely, we show that, if $G$ has a fixed planar embedding, $H$ is biconnected, and the undrawn part of~$G$ only consists of $p$ length-$2$ paths, then \PDE can be solved in $O(p^2n)$ time.
Furthermore, we show that \PDE is FPT with respect to the vertex cover number of $G$, both at fixed and at variable embedding. Our proofs rely on the following techniques. In the fixed embedding setting, we show that the problem admits a kernel whose size is linear in the vertex cover number of $G$ and that it can be modeled as a first-order formula in the Existential Theory of the Reals (ETR), whose satisfiability can be tested in exponential time using Renegar's decision algorithm for ETR~\cite{Renegar92one,Renegar92two,Renegar92three}.
In the variable embedding setting, we show that the vertices of $G - H$ can be reduced to a subset having a size that is linear in the vertex cover number $k$ of $G$, while preserving equivalence with the original instance, and that there exists a subset of $O(k)$ ``interesting'' faces of $\Gamma_H$ that can be used to host the vertices in $G-H$. These facts, combined with an ETR formulation for the variable embedding setting, yield fixed-parameter tractability.
Full proofs of all the statements can be found in the full version of the paper.


\section{Preliminaries}\label{se:preliminaries}

In this section, we give some preliminaries and basic definitions. For standard concepts and terminology about graphs and their drawings, we refer the reader to~\cite{BattistaETT99,0030488}.


An \emph{embedded graph} is a planar graph with a fixed planar embedding, which determines the rotation scheme at each vertex.
For an embedded graph with planar embedding $\cal E$, an \emph{embedding-preserving drawing} is a drawing that \emph{respects} $\cal E$, i.e., a planar drawing that belongs to $\cal E$. For two points $u$ and $v$ in the plane, or two vertices $u$ and $v$ drawn in the plane, $\ell_{uv}$ denotes the straight line passing through $u$ and $v$, oriented so that $u$ is encountered before $v$.

\subparagraph*{Partial drawing extensibility.}
Let $G$ be a planar graph, $H$ be a subgraph of $G$, and $\Gamma_H$ be a straight-line planar drawing of $H$. A straight-line planar drawing $\Gamma$ of $G$ \emph{extends} $\Gamma_H$ if each vertex of $H$ is assigned to the same point in $\Gamma$ as in $\Gamma_H$. If a straight-line planar drawing of $G$ that extends $\Gamma_H$ exists, then the triple $\langle G,H,\Gamma_H\rangle$ is \emph{extensible}. We can assume that $H$ is an induced subgraph of $G$, as the drawing of the edges in $E(G)\setminus E(H)$ with both end-vertices in $H$ is determined by $\Gamma_H$. Next, we define the problems we study.

\medskip
\noindent{\sc Partial Drawing Extensibility} (\PDE): Given a planar graph $G$, a subgraph $H$ of $G$, and a straight-line planar drawing $\Gamma_H$ of $H$, is $\langle G,H,\Gamma_H\rangle$ extensible?

\medskip
\noindent{\sc Partial Drawing Extensibility in an Internal Face} (\PDEinF): Given a planar graph $G$, a cycle $C$ of $G$, and a simple polygon $\Gamma_C$ representing $C$, does a straight-line planar drawing of $G$ that extends $\Gamma_C$ in which the vertices of $G - C$ lie in the interior of $\Gamma_C$ exist?

\medskip
\noindent{\sc Partial Drawing Extensibility in the Outer Face} (\PDEoutF): Given a planar graph $G$, a cycle $C$ of $G$, and a simple polygon $\Gamma_C$ representing $C$, does a straight-line planar drawing of $G$ that extends $\Gamma_C$ in which the vertices of $G - C$ lie in the exterior of $\Gamma_C$ exist?





If $G$ has a prescribed planar embedding $\cal E$, then \PDE additionally requires that the sought drawing of $G$ extending $\Gamma_H$ respects~$\cal E$.

%

A function~$f:\mathbb{N}^+\rightarrow \mathbb{R}_{\geq 0}$ is \emph{super-additive} if~$f(\sum_i x_i)\geq \sum_i f(x_i)$.

\begin{restatable}{lemma}{lemmafixedembfromfacetograph} \label{le:fixed-emb-from-face-to-graph}
Let $\cal I$ be a family of instances for the \PDE problem such that $H$ is biconnected and $G$ has a fixed planar embedding. Let ${\cal I}_f$ be the family of instances of the
\PDEinF and \PDEoutF problems induced by restricting the instances $\langle G, H,\Gamma_H \rangle \in \cal I$ so that $H$ is a facial cycle of $\Gamma_H$. Also,  let $A_{in}$ and $A_{out}$ be algorithms that solve \PDEinF and \PDEoutF on any instance $I_f \in {\cal I}_f$ in $O(f_{in}(|I_f|))$ and $O(f_{out}(|I_f|))$ time, respectively, such that $f_{in}, f_{out}: \mathbb{N}\rightarrow \mathbb{N}$ are super-additive functions. Then, the \PDE problem can be solved in $O\big(\max\{f_{in}(n),f_{out}(n)\}\big)$ time for any $n$-vertex instance  $I \in {\cal I}$.
\end{restatable}

\medskip
A \emph{$2$-path} is a path of length $2$. A \emph{\pt instance} $\langle G,H,\Gamma_H\rangle$ of \PDE is such that $H$ is biconnected and the subgraph of $G$ to be drawn consists of a set of $2$-paths, that are vertex-disjoint except, possibly, at their end-vertices. More precisely, the edges in $E(G)-E(H)$ induce a set $\mathcal P$ of $2$-paths such that: (i) the extremes of each path in $\mathcal P$ are in $H$ and might be shared by different paths; and (ii) the intermediate vertex of each path in $\mathcal P$ is not in $H$ and belongs to a single path in $\mathcal P$; refer to \cref{fig:p2-pendant}.
Similarly, in a \emph{\pt instance} of \PDEinF or \PDEoutF, the subgraph of $G$ to be drawn consists of a set of $2$-paths.

\begin{figure}
    \centering
    \includegraphics[width=0.5\linewidth]{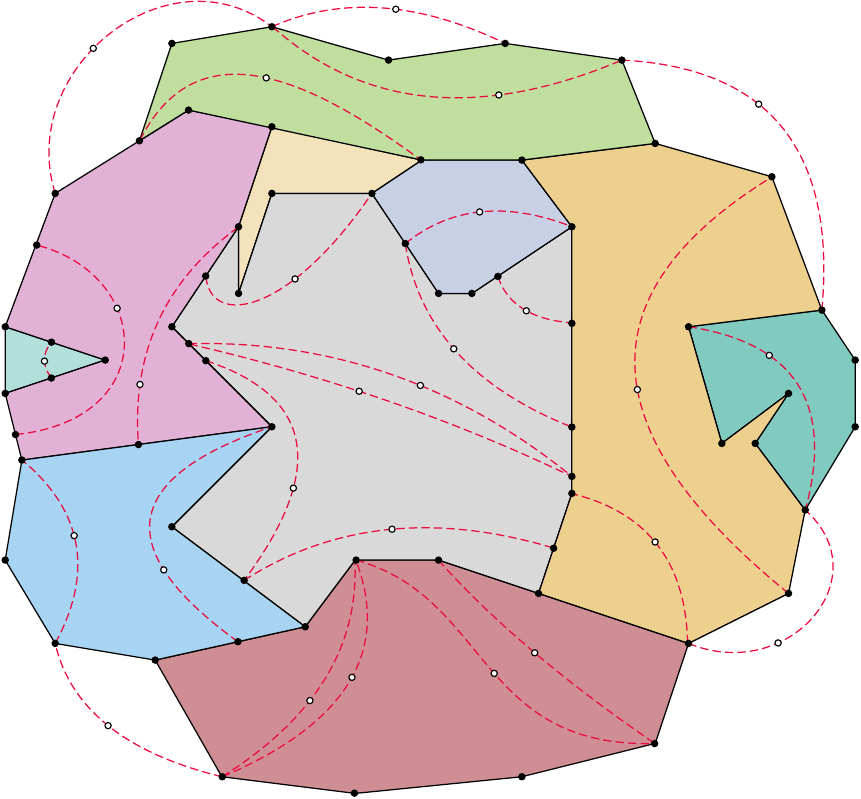}
    \caption{A \pt instance of \PDE. Edges in $E(G)-E(H)$ are represented by dashed lines.}
    \label{fig:p2-pendant}
\end{figure}

\subparagraph{Fixed-parameter tractability.} 
A problem~$\mathcal P$ taking as input an $n$-vertex graph $G$ is said to be \emph{fixed-parameter tractable} (FPT) with respect to a parameter $k$ if there exists an algorithm solving it in time $O(f(k)\cdot poly(n))$, where $f$ is a computable function depending only on $k$.
A \emph{kernelization algorithm} for~$\mathcal P$ is a polynomial-time (in~$n$) procedure that transforms an instance $\langle G,k\rangle$ into an equivalent instance $\langle G’,k’\rangle$, called \emph{kernel}, such that the \emph{size} of the kernel, i.e., the number of vertices of $G’$, and the parameter $k’$ are bounded by computable functions of $k$. The existence of such kernelization implies that $\mathcal P$ is FPT when parameterized by $k$, if the kernel can be solved by a brute-force procedure.
For problems involving geometric aspects, like ours, 
explicit exponential-time solutions for the kernel might be needed.

\subparagraph{Existential theory of the reals.} An \emph{ETR formula} is a first-order existential sentence about the reals, that is, a first-order formula of the form  $\exists x_1,\dots,\exists x_n \Phi(x_1,\dots,x_n)$ in prenex normal form where the prefix only consists of existential quantifiers and $\Phi$ is a quantifier‑free Boolean combination of polynomial equalities/inequalities with rational (or algebraic) coefficients. 
\removedforshortening{The Existential Theory of the Reals problem asks whether an ETR formula is true over $\mathbb{R}$. This decision problem is \NP-hard and lies in PSPACE.  The class $\exists\mathbb{R}$ consists of the decision problems that reduce in polynomial time to ETR. The PSPACE upper bound comes from Canny’s celebrated result~\cite{Canny88,Schaefer2013}.
We will make use of Renegar's decision algorithm for ETR.}

\begin{theorem}[Renegar~\cite{Renegar92one,Renegar92two,Renegar92three}]
  \label{thm:renegar}
  One can decide if an existential sentence $\Phi$ is true or false in time $(L\log L\log\log L) \cdot (P\cdot R)^{O(N)}$,
  where~$N$ is the number of variables in~$\Phi$,~$P$ is the
  number of polynomials in~$\Phi$, $R$ is the maximum total
  degree of the polynomials in~$\Phi$, and~$L$ is the maximum length
  of the binary representation of the coefficients of the polynomials
  in~$\Phi$.
\end{theorem}





\section{Complexity Results for Biconnected \PDE}\label{se:hardness}


\begin{figure}[tb]
    \begin{subfigure}[b]{0.25\textwidth}
      \centering
      \includegraphics[width=1\textwidth,page=1]{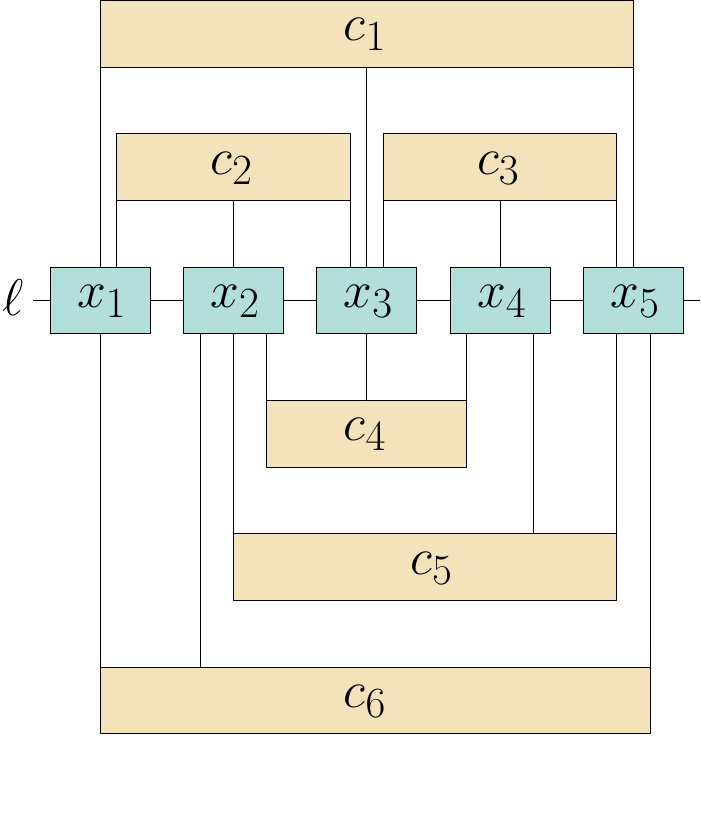}
      \subcaption{}
      \label{fig:variable-clause-graph-drawing-a}
    \end{subfigure}
    \hfill
    \begin{subfigure}[b]{0.70\textwidth}
      \centering
      \includegraphics[width=1\textwidth,page=1]{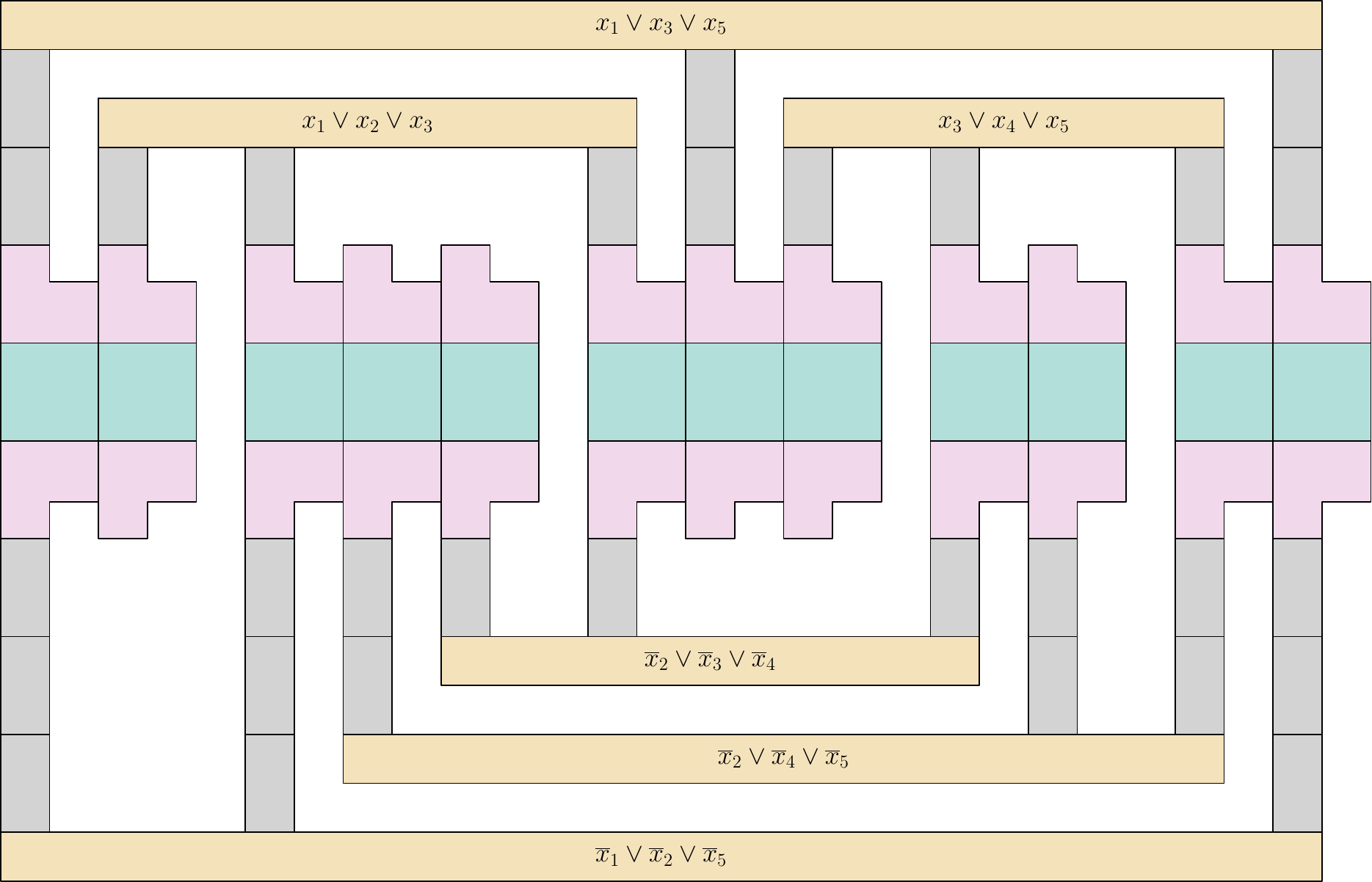}
      \subcaption{}
      \label{fig:variable-clause-graph-drawing-b}
    \end{subfigure}
    \caption{(a) The monotone rectilinear representation $\Gamma_\varphi$ of graph $G_\varphi$ of the instance of \monthreesat $\varphi = c_1 \wedge c_2 \wedge c_3 \wedge c_4 \wedge c_5$ where $c_1 = (x_1 \vee x_3 \vee x_5)$, $c_2 = (x_1 \vee x_2 \vee x_3)$, $c_3 = (x_3 \vee x_4 \vee x_5)$, $c_4 = (\overline{x}_2 \vee \overline{x}_3 \vee \overline{x}_4)$, $c_5 = (\overline{x}_2 \vee \overline{x}_4 \vee \overline{x}_5)$, and $c_6 = (\overline{x}_1 \vee \overline{x}_2 \vee \overline{x}_5)$. (b) A boxed drawing $\overline{\Gamma}_\varphi$ of $G_\varphi$.}
    \label{fig:variable-clause-graph-drawing}
\end{figure}

\begin{figure}[tb]
    \hfill
        \begin{subfigure}[c]{0.22\textwidth}
      \centering
      \includegraphics[height=3\textwidth,page=9]{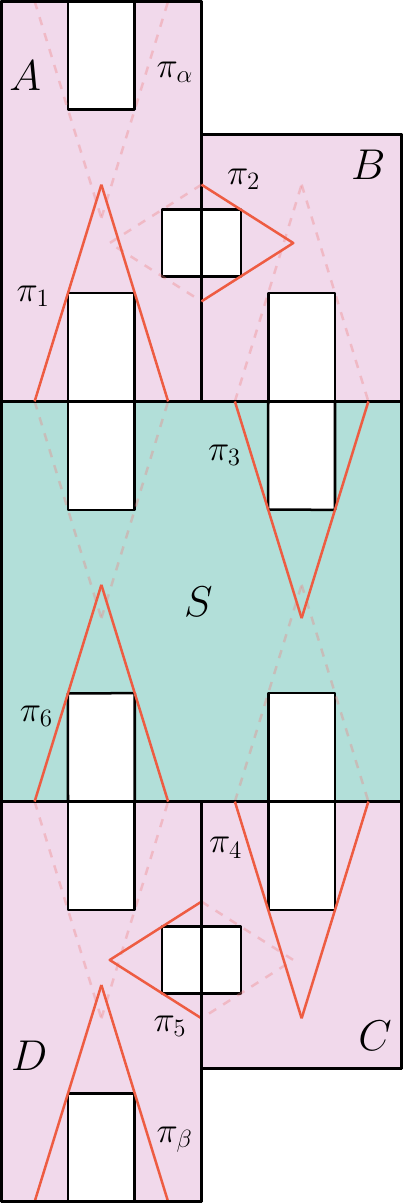}
      \subcaption{}
      \label{fig:variable-gadget-00}
    \end{subfigure}
    \hfill
        \begin{subfigure}[c]{0.22\textwidth}
      \centering
      \includegraphics[height=3\textwidth,page=10]{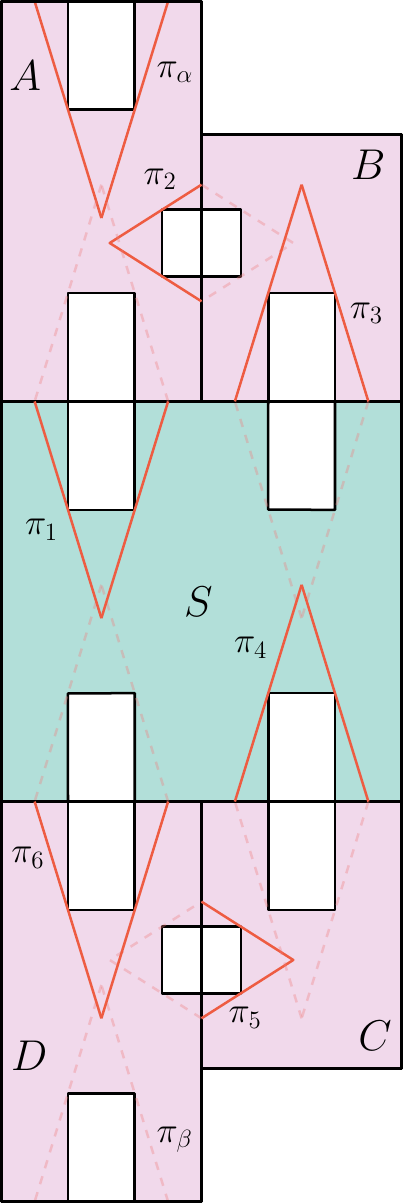}
      \subcaption{}
      \label{fig:variable-gadget-11}
    \end{subfigure}
        \hfill
        \begin{subfigure}[c]{0.22\textwidth}
      \centering
      \includegraphics[height=3\textwidth,page=11]{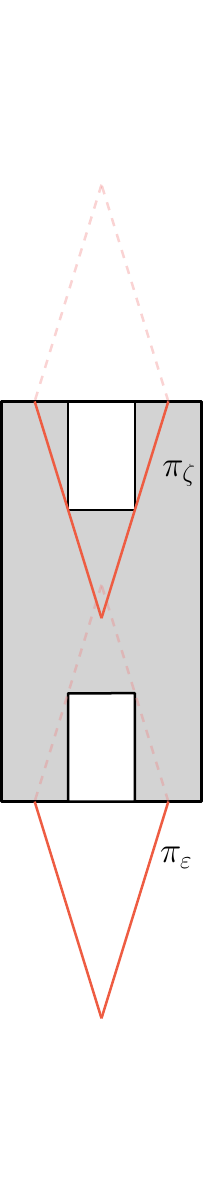}
      \subcaption{}
      \label{fig:variable-transmission-a}
    \end{subfigure}
    \hfill
        \begin{subfigure}[c]{0.22\textwidth}
      \centering
      \includegraphics[height=3\textwidth,page=12]{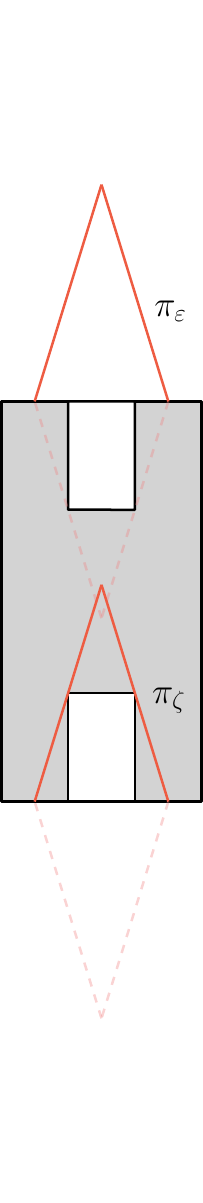}
      \subcaption{}
      \label{fig:variable-transmission-b}
    \end{subfigure}
    \hfill~
    \caption{(a)-(b) \false and \true configurations for a vertical-structure of a variable-gadget. (c)-(d) Two drawings of a transmission-box of a transmission-gadget.}
    \label{fig:variable-gadget-proof}
\end{figure}

\begin{figure}[tb]
    \begin{subfigure}[b]{0.45\textwidth}
      \centering
      \includegraphics[width=1\textwidth,page=2]{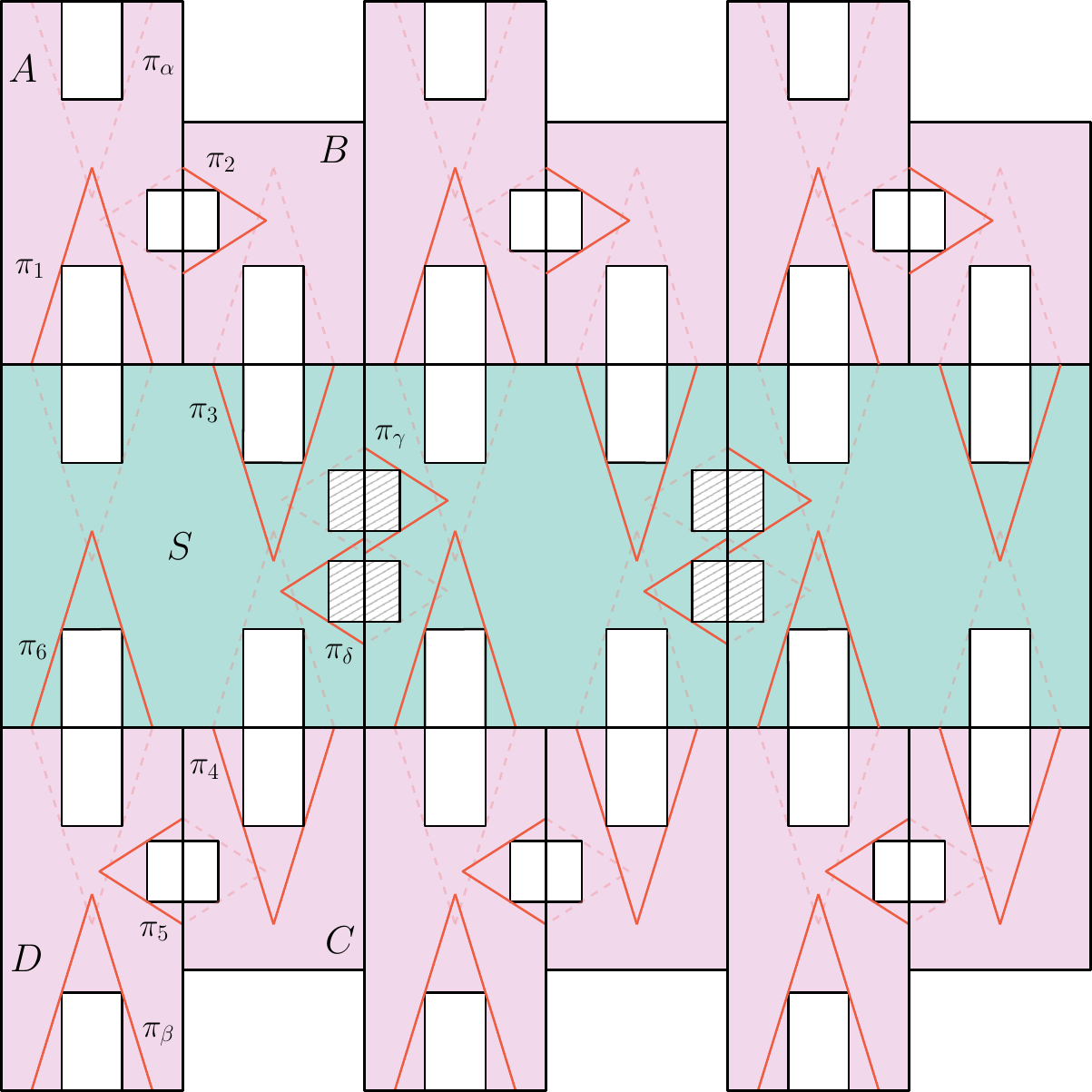}
      \subcaption{}
      \label{fig:variable-gadget-a}
    \end{subfigure}
    \hfill
    \begin{subfigure}[b]{0.45\textwidth}
      \centering
      \includegraphics[width=1\textwidth,page=3]{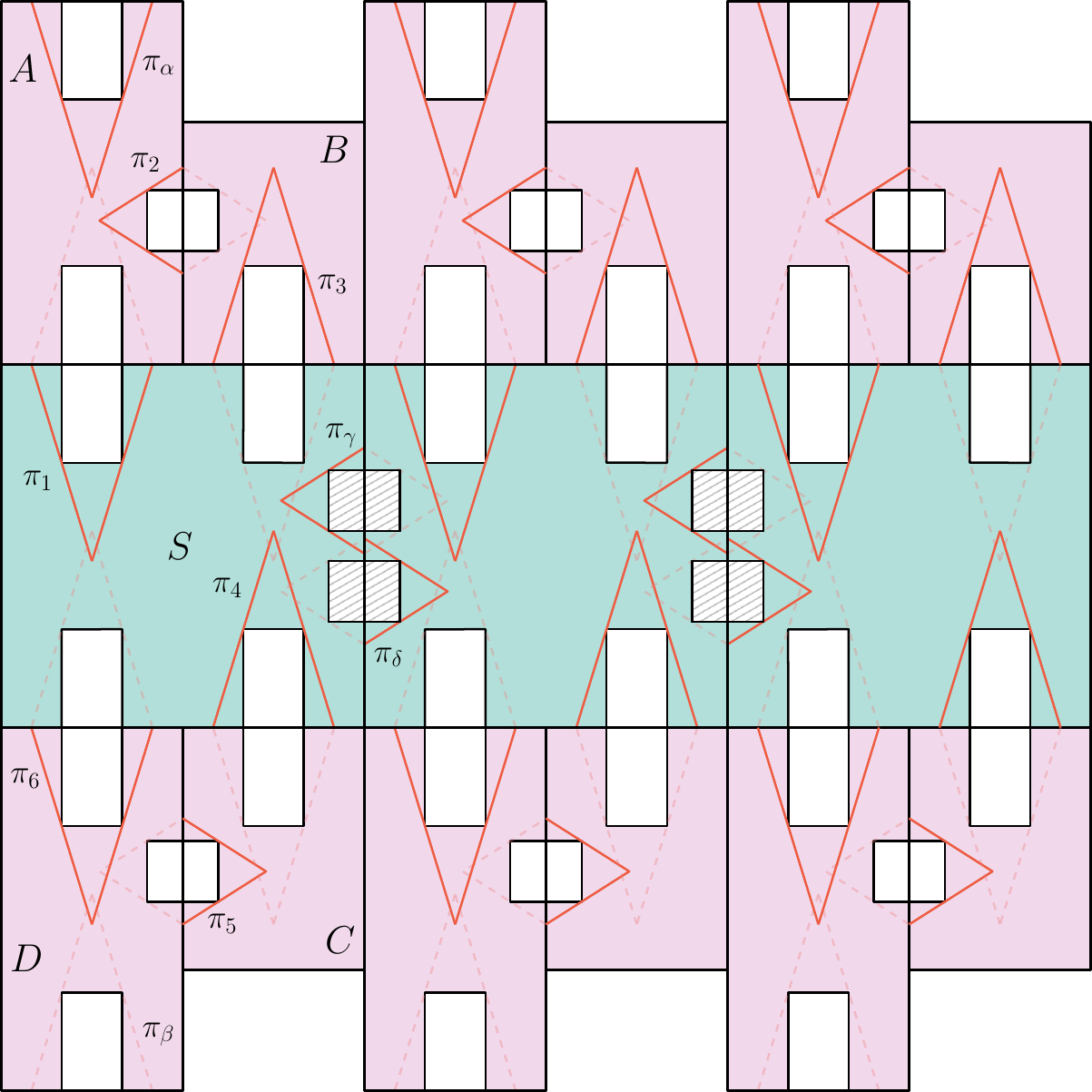}
      \subcaption{}
      \label{fig:variable-gadget-b}
    \end{subfigure}
    \hfill~
    \caption{A negative drawing (a) and a positive drawing (b) of the variable-gadget.}
    \label{fig:variable-gadget}
\end{figure}


In this section, we study the complexity of \PDE for instances such that $H$ is biconnected.

\begin{restatable}{theorem}{thHardness}\label{th:hardness}
\PDE is \NP-hard even for \pt instances $\langle G,H,\Gamma_H\rangle$ such that $G$ is subcubic, $H$ is biconnected, and $\Gamma_H$ has faces of bounded size.
\end{restatable}

We prove \cref{th:hardness} via a reduction from the \NP-complete problem \monthreesat~\cite{DBLP:journals/ijcga/BergK12}, a variant of \threesat where the Boolean formula $\varphi$ has clauses composed of three positive or three negated literals and the incidence graph $G_\varphi$ of the variables and clauses of $\varphi$ comes with a planar drawing $\Gamma_\varphi$ where all the variables lay along a straight line $\ell$, and clauses with directed (negated) literals lay above (resp.\ below) $\ell$; see, for example, \cref{fig:variable-clause-graph-drawing-a}.


Given an instance $\varphi$ of \monthreesat, we construct an instance $\langle G,H,\Gamma_H\rangle$ of \PDE that is extensible if and only if $\varphi$ is satisfiable.
First, starting from $\Gamma_\varphi$, we construct a \emph{boxed drawing} $\overline{\Gamma}_\varphi$ of $G_\varphi$ as follows (see \cref{fig:variable-clause-graph-drawing-b}). For each variable~$x_i$, $i=1, \dots, n$, let $m_i$ be the maximum between the number of directed literals $x_i$ and the number of negated literals $\overline{x}_i$ in $\varphi$. The variable vertex corresponding to $x_i$ in $\Gamma_\varphi$ is replaced in $\overline{\Gamma}_\varphi$ by a sequence of $m_i$ axis-aligned \emph{vertical-structures}, each composed by a rectangular \emph{variable-box} (green in \cref{fig:variable-clause-graph-drawing-b}) and two L-shaped polygons, called \emph{foot-boxes} (pink in \cref{fig:variable-clause-graph-drawing-b}). Each clause vertex is replaced by a rectangular box, called \emph{clause-box} (yellow in \cref{fig:variable-clause-graph-drawing-b}). Each edge between a variable vertex and a clause vertex is replaced by a vertical sequence of axis-aligned boxes, called \emph{transmission-boxes}  (gray in \cref{fig:variable-clause-graph-drawing-b}). 
%
Instance $\langle G,H,\Gamma_H\rangle$ is obtained from~$\overline{\Gamma}_\varphi$ via the following construction, where any object introduced in $\Gamma_H$ is implicitly\rfs{regarded as being} introduced also in $H$ and $G$, while any object introduced in $G$ is assumed to not belong to $H$ (nor to $\Gamma_H$).
First, for each point in common between two or more straight-line segments of $\overline{\Gamma}_\varphi$ we introduce a vertex in $\Gamma_H$ with the same coordinates. Second, for each segment of~$\overline{\Gamma}_\varphi$ we introduce a straight-line edge in $\Gamma_H$ between the endvertices corresponding to the endpoints of the segment.
Third, for each box of $\overline{\Gamma}_\varphi$ we introduce a set of objects (vertices and edges) described below, depending on the type of the box. This produces variable-gadgets (corresponding to several vertical-structures), transmission-gadgets (corresponding to one or more transmission-boxes), and clause-gadgets (each corresponding to a single clause-box). 

\subparagraph{Variable-gadget.} 

For each vertical-structure of variable $x_i$, we add to $\Gamma_H$ the rectangular obstacles filled white in \cref{fig:variable-gadget-00} and add to~$G$ the $2$-paths $\pi_1, \pi_2, \dots, \pi_6, \pi_\alpha,\pi_\beta$ depicted with red lines in the same figure. The size and position of the rectangular obstacles and the attachment points of the $2$-paths can be defined so that some pairs of $2$-paths cannot simultaneously lie inside the same rectangle without intersecting. 
In particular, 
in any straight-line planar drawing of the vertical-structure extending $\Gamma_H$, the faces into which the $2$-paths lay correspond to either the configuration in \cref{fig:variable-gadget-00} (\emph{negative} drawing of the vertical-structure encoding a \false value) or the configuration in  \cref{fig:variable-gadget-11} (\emph{positive} drawing encoding a \true value). 
%
%
When two vertical-structures of the same variable-gadget are consecutive (i.e., they share a vertical path, see \cref{fig:variable-gadget}), we introduce some obstacles and two $2$-paths so that all vertical-structures of the same variable-gadget encode the same truth value.

\subparagraph{Transmission-gadget.} 
The transmission-gadget is composed of a vertical sequence of transmission-boxes each sharing one horizontal side with the next (refer to \cref{fig:variable-transmission-a,fig:variable-transmission-b}). Each transmission-box contains two obstacles and two $2$-paths $\pi_\varepsilon$ and $\pi_\zeta$. The size and position of the rectangular obstacles and the attachment points of the $2$-paths can be defined in such a way that these paths are incompatible in the rectangle of the transmission-box. When two transmission-boxes share a horizontal side, the path $\pi_\varepsilon$ of the box above is identified with the path $\pi_\zeta$ of the other box below.
Since the transmission-gadget is attached to a variable-gadget we have that a \true value (\false value, resp.) of the variable-gadget is necessarily transmitted along the transmission-gadget to the clause-gadget corresponding to a negative clause below (to a positive clause above, resp.). 


\subparagraph{Clause-gadget.}
We describe the clause-gadget $C$ for a positive clause $c = (x_i \vee x_j \vee x_k)$; the one for a negative clause is similar. 
It is composed of a single clause-box inside of which we introduce three obstacles and five $2$-paths $\pi_{i}$, $\pi_{j}$, $\pi_{k}$, $\pi_{i,j}$, and $\pi_{i,j,k}$ depicted in \cref{fig:clause}. Each of $\pi_{i}$, $\pi_{j}$, and $\pi_{k}$ is identified with the path $\pi_\varepsilon$ of the transmission-gadget coming from the variable-gadget of the variables $x_i$, $x_j$, and $x_k$, respectively. 
Also, we consider the five notable points $P_i$, $P_j$, $P_k$, $P_{i,j}$, and $P_{i,j,k}$ in \cref{fig:clause}. 
The size and position of the rectangular obstacles and the attachment points of the $2$-paths can be defined so that when $\pi_{i}$ ($\pi_j$, $\pi_k$, resp.) is inside the clause-box, the closed region bounded by such a path and the bottom side of the clause-box encloses $P_i$ ($P_j$, $P_k$, resp.). This property implies that when both $\pi_i$ and $\pi_j$ are inside the clause-box, $\pi_{i,j}$ encloses $P_{i,j}$. It follows that if $\pi_i$, $\pi_j$, and $\pi_k$ are inside the clause-box, then $\pi_{i,j,k}$ encloses $P_{i,j,k}$ and thus intersects the clause-box. Conversely, if at least one between $\pi_i$, $\pi_j$ or $\pi_k$ is outside the clause-box, then it is possible to draw the remaining paths, $\pi_{i,j}$, and $\pi_{i,j,k}$ inside the clause-box without any intersection. 



\begin{figure}[tb!]
    \centering
    \includegraphics[width=\linewidth,page=6]{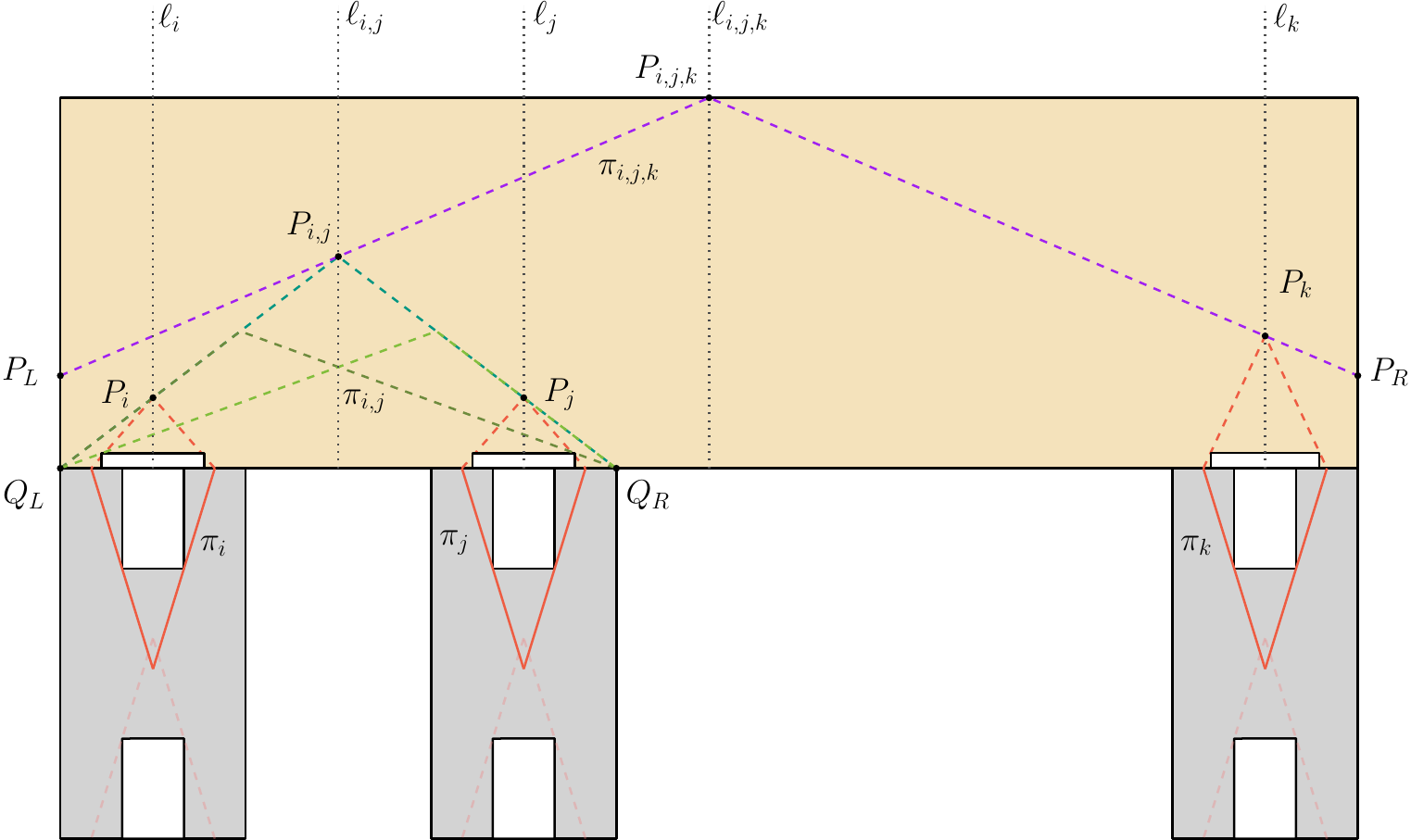}
    \caption{Clause gadget. Distinct drawings of $\pi_{i,j}$ are dashed and have different shades of green. 
    }
    \label{fig:clause}
\end{figure}


\subparagraph{Sketch of the proof of \cref{th:hardness}.}
The instance $\langle G, H, \Gamma_H\rangle$ of \PDE can be constructed in polynomial time and is equivalent to the formula $\varphi$ of \monthreesat. The maximum degree of the vertices in $\langle G, H, \Gamma_H\rangle$ is four, however it can be reduced to three by suitably replacing degree-$4$ vertices with sufficiently small constructions. Finally, the internal faces of $\Gamma_H$ have bounded size, whereas the outer face does not. However, since the outer face plays no role in the reduction, additional vertices and edges can be inserted into it to ensure that all faces have bounded size.
\qed


\medskip\noindent
Two immediate consequences of  \cref{th:hardness} are the following.

\begin{restatable}{corollary}{coETH}\label{co:eth}
Unless the Exponential Time Hypothesis (ETH) fails, \PDE has no~$2^{o(\sqrt{n})}$-time and no $2^{o(tw)}$-time algorithm even for \pt instances, where $n$ and $tw$ are the number of vertices and the tree-width of the input graph, respectively. 
\end{restatable}

\begin{restatable}{corollary}{coHardnessTrees}\label{cor:hardness-trees}
PDE is \NP-hard even for subcubic instances $\langle G,H,\Gamma_H\rangle$ such that $H$ is a tree.
\end{restatable}


\section{P$_2$-Pendant Instances}\label{se:pendant}

In this section we discuss the positive results when the instance is restricted to be \pt. We have the following two theorems.

\begin{theorem}\label{th:pde-in-f}
\PDEinF can be solved in $O(p^2 n)$ time for $n$-vertex \pt instances such that the set $\mathcal P$ of $2$-paths contains $p$ paths.
\end{theorem}

\begin{restatable}{theorem}{pdeoutf}\label{th:pde-out-f}
\PDEoutF can be solved in $O(p^2 n)$ time for $n$-vertex \pt instances such that the set $\mathcal P$ of $2$-paths contains $p$ paths.
\end{restatable}

Note that, in the instances of \PDEinF and \PDEoutF, the input graph $G$ does not have a fixed planar embedding. However, as we will show soon, a planar embedding for $G$ can be fixed without loss of generality. In light of this, combining \cref{th:pde-in-f,th:pde-out-f} with \cref{le:fixed-emb-from-face-to-graph}, we obtain the following main result.

\begin{theorem}\label{th:p2-fixed-emb}
\PDE can be solved at fixed embedding in $O(p^2 n)$ time for $n$-vertex \pt instances such that the set $\mathcal P$ of $2$-paths contains $p$ paths.
\end{theorem}

The rest of the section is mainly devoted to a proof of \cref{th:pde-in-f}. At the end of the section we will briefly discuss the main differences in the proof of \cref{th:pde-out-f} with respect to the one of \cref{th:pde-in-f}. The general strategy of the algorithms for such proofs is the same, and is described below. 



\subparagraph{General strategy.} Let $\langle G,C,\Gamma_C \rangle$ be a \pt instance of \PDEinF or \PDEoutF, see \cref{fig:oriented-dual}. Recall that $C$ is a cycle and $\Gamma_C$ is a polygon representing $C$. Let $\mathcal P$ be the set of $2$-paths induced by $E(G)-E(H)$. We can assume that no two paths in $\mathcal P$ share both their end-vertices, as otherwise the instance can be reduced in $O(p^2)$ time to an equivalent instance that contains only one of the multiple paths sharing their end-vertices. Indeed, in a solution for the reduced instance, the removed paths can be inserted sufficiently close to the path with the same end-vertices so that no crossing is introduced.
We construct an embedding of $G$ by inserting every path in $\mathcal P$ inside $C$ for \PDEinF or outside $C$ for \PDEoutF; the rotation system at each vertex $v$ is determined by the circular order in $C$ of the end-vertices of such paths different from~$v$. Note that this order is well-defined because of the assumption that no two paths in $\mathcal P$ share both their end-vertices.
We henceforth treat $G$ 
as an embedded graph. 
Every solution to $\langle G,C,\Gamma_C \rangle$ is a straight-line planar drawing of the embedded graph $G$ in which the outer face, for \PDEinF, or an internal face, for \PDEoutF, is delimited by $C$. Note that $G$ is a subdivision of an outerplanar graph~$O$, whose outer face is bounded by $C$ and whose internal edges correspond to the paths in $\cal P$. Whether the embedding of $G$ is planar can be tested in $O(n)$ time by verifying whether~$O$ is outerplanar \cite{DBLP:journals/ipl/Mitchell79,syslo-outerplanarity-testing}. We arbitrarily orient each internal edge of $O$. 
We construct the weak dual $T$ of $O$. See \cref{fig:oriented-dual-internal,fig:oriented-dual-external,fig:oriented-dual-external-2}. By looking at the structure of $T$ and at the geometry of $\Gamma_C$, we will find a path $P_e$ in $\mathcal P$ that can be used to split the instance $\langle G,C,\Gamma_C \rangle$ into smaller instances on which recursion can be applied. Then $\langle G,C,\Gamma_C \rangle$ is a positive instance if and only if all the smaller instances are. Next, we formalize this argument for \PDEinF and thus provide a proof of \cref{th:pde-in-f}. 

\subparagraph{Orientation of $\mathbf{T}$.} We consider and possibly orient each edge $e$ of $T$ as follows. Let $(u,v)$ be the dual edge of $e$, which is an internal edge of $O$. Assume w.l.o.g.\ that $(u,v)$ is oriented from $u$ to $v$, and let $P_e=(u,w,v)$ be the path in $\mathcal P$ corresponding to $(u,v)$. Slightly overloading the term, we also say that $P_e$ corresponds to $e$. If the straight-line segment $\overline{uv}$ lies inside, then $e$ is not oriented. This condition can be tested efficiently.

\begin{figure}
    \centering
    \begin{subfigure}[t]{0.25\textwidth}
       \includegraphics[width=\textwidth, page=1]{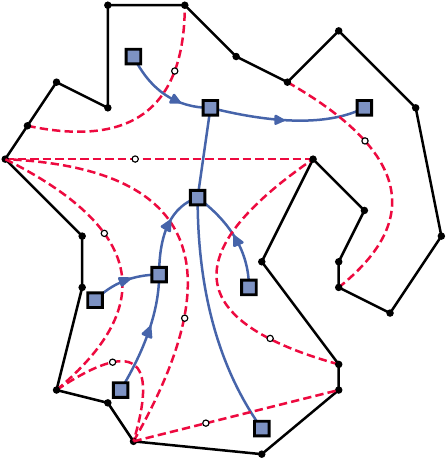}
        \subcaption{}
        \label{fig:oriented-dual-internal} 
    \end{subfigure}
    \hfill
    \begin{subfigure}[t]{0.3\textwidth}
        \includegraphics[width=\textwidth, page=4]{figures/dual.pdf}
        \subcaption{}
        \label{fig:oriented-dual-external} 
    \end{subfigure}
    \hfill
    \begin{subfigure}[t]{0.3\textwidth}
        \includegraphics[width=\textwidth, page=5]{figures/dual.pdf}
        \subcaption{}
        \label{fig:oriented-dual-external-2} 
    \end{subfigure}
    \caption{Instances of \PDEinF (a) and \PDEoutF (b) and (c), together with their oriented duals. In (b) and (c) the same instance is depicted with a different choice for the outer face. }
    \label{fig:oriented-dual} 
    
\end{figure}

\begin{restatable}{lemma}{lemmainsidedetection} \label{le:inside-detection}
It can be tested in $O(n)$ time whether $\overline{uv}$ lies inside $\Gamma_C$.
\end{restatable}




If $\overline{uv}$ does not lie inside $\Gamma_C$, we have that $P_e$ can be embedded in at most one of the two half-planes defined by $\ell_{uv}$. 

\begin{restatable}{lemma}{lemmadonut} \label{le:donut}
Suppose that $\overline{uv}$ does not lie inside $\Gamma_C$. Then, every straight-line drawing of $P_e$ inside $\Gamma_C$ lies in the same half-plane defined by $\ell_{uv}$.
\end{restatable}

In the following, we show how to decide the half-plane in which $P_e$ can be embedded, if any, assuming that $\overline{uv}$ does not lie inside $\Gamma_C$. A \emph{proper left} (\emph{right}) \emph{drawing} of $P_e$ is a straight-line drawing of $P_e$ that lies to the left (resp.\ right) of $\ell_{uv}$ (except at $u$ and $v$) and inside $\Gamma_C$ (except at $u$ and $v$), see \cref{fig:proper-left}. 
We show how to test whether $P_e$ admits a proper left drawing; the test for a proper right drawing is symmetric. Let $\mathcal L$ be the open half-plane to the left of $\ell_{uv}$.
Let $Q_1$ and $Q_2$ be the paths composing $C$ that are traversed in counter-clockwise and clockwise direction from $u$ to $v$, respectively. 
%
%
%
Let $\mathcal L^+$ be the region $\mathcal L\cup \overline{uv}$; that is, the half-plane delimited by $\ell_{uv}$, closed at $\overline{uv}$ but not at $\ell_{uv}-\overline{uv}$. The intersection $Q_1(\mathcal L)$ between the polygonal line representing $Q_1$ in $\Gamma_C$ and $\mathcal L^+$ consists of a set of polygonal chains (or just \emph{chains}). We say that a chain is \emph{bad} if it has one point in the interior of the segment $\overline{uv}$ and one endpoint in $\ell_{uv}-\overline{uv}$, see \cref{fig:bad-chains}. If $Q_1(\mathcal L)$ contains a bad chain, then $P_e$ does not admit a proper left drawing.
\begin{figure}[htb]
    \centering
     \begin{subfigure}[t]{0.48\textwidth}
        \includegraphics[width=0.99\linewidth,page=1]{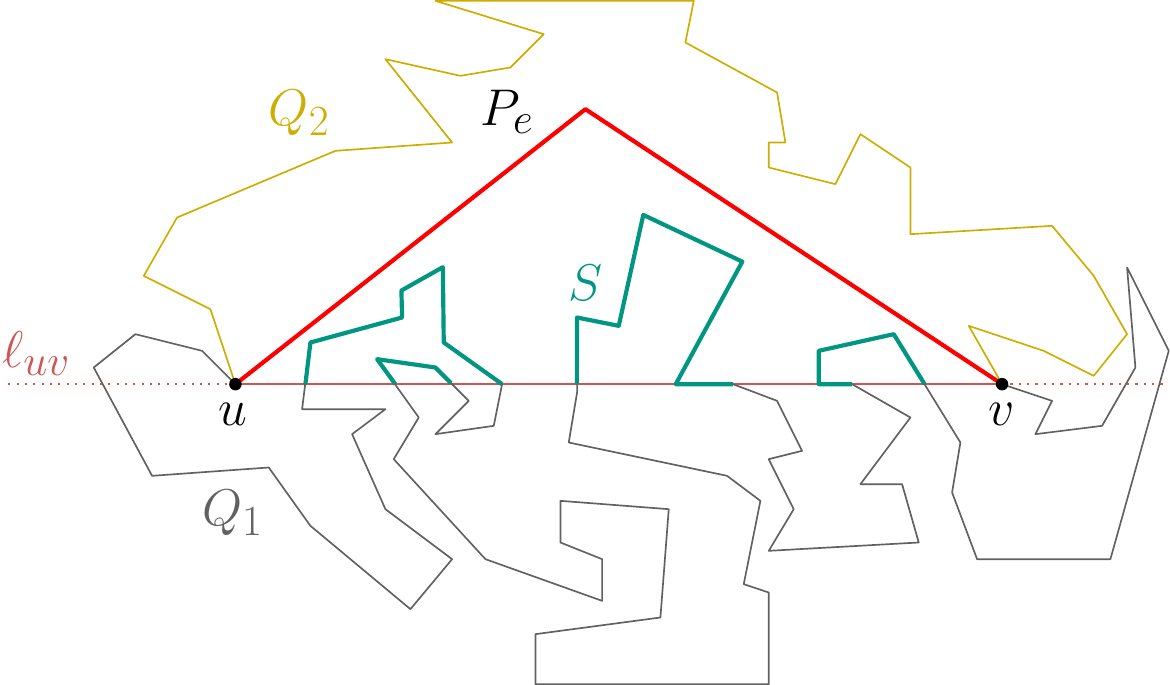}
        \subcaption{}
        \label{fig:proper-left} 
    \end{subfigure}
    \hfil
     \begin{subfigure}[t]{0.48\textwidth}
        \includegraphics[width=0.99\linewidth,page=2]{figures/chains.pdf}
        \subcaption{}
        \label{fig:leaning-representation} 
    \end{subfigure}
    \caption{(a) A proper left drawing of $P_e$, thick and red. Chains in $S$ are thick and green. (b) The leaning representation $\Gamma_e$ of $P_e$, thick and orange. Different colors in the regions defined by $\Gamma_e$ show the interior of the polygons in different instances on which recursion is applied.}
    \label{fig:proper-leaning}
\end{figure}
%
%
\begin{figure}[htb]
    \centering
        \includegraphics[width=0.4\linewidth,page=4]{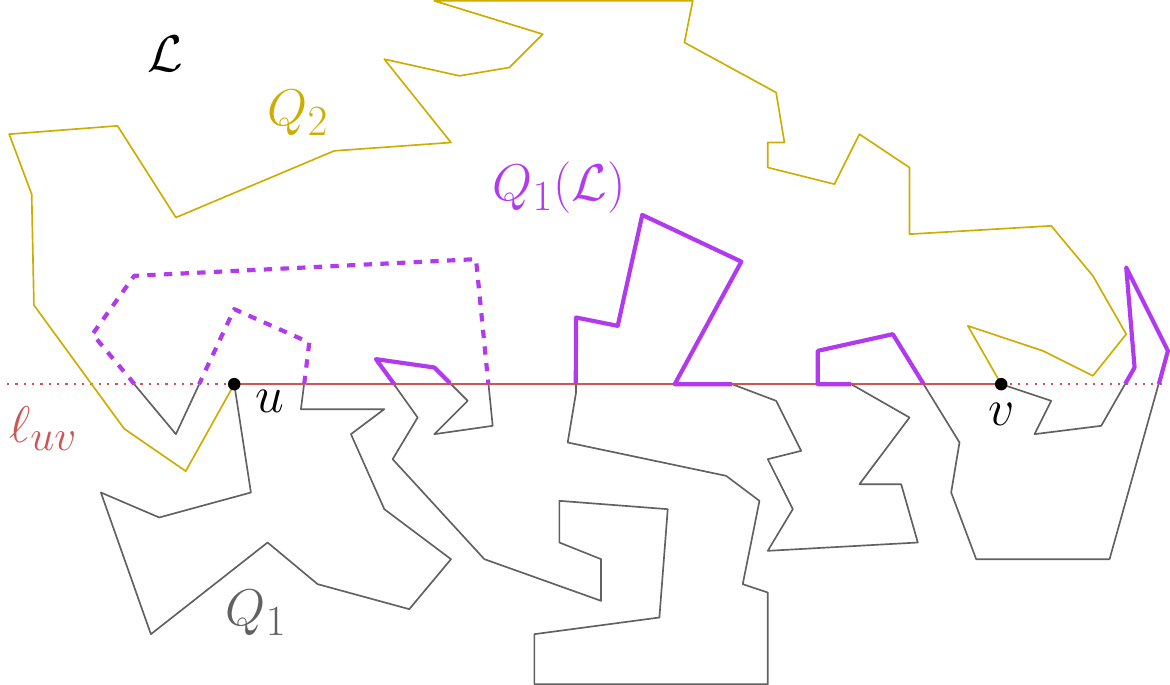}

    \caption{The polygonal chains in $Q_1(\mathcal L)$ are violet. The two chains that are dashed are bad.}
    \label{fig:bad-chains} 
\end{figure}
%
Otherwise, let $S$ be the set of chains that have both their endpoints on $\overline{uv}$, possibly at $u$ or $v$. If $S$ is empty, then we have that $P_e$ does not admit a proper left drawing. 
If $S$ is non-empty, let $t_u$ ($t_v$) be the half-line with the following properties: (i) it starts at $u$ (resp.\ at $v$); (ii) it passes through a vertex different from $u$ (resp.\ from~$v$) of a chain in $S$; and (iii) the angle defined by a counter-clockwise (resp.\ clockwise) rotation that brings $\ell_{uv}$ to overlap with $t_u$ (resp.\ $t_v$) is maximum, subject to being smaller than $\pi$.
If $t_u$ and $t_v$ do not intersect, as in \cref{fig:t_u-t_v no intersection}, then we have that $P_e$ does not admit a proper left drawing.
\begin{figure}[htb]
    \centering
        \includegraphics[width=0.4\linewidth,page=5]{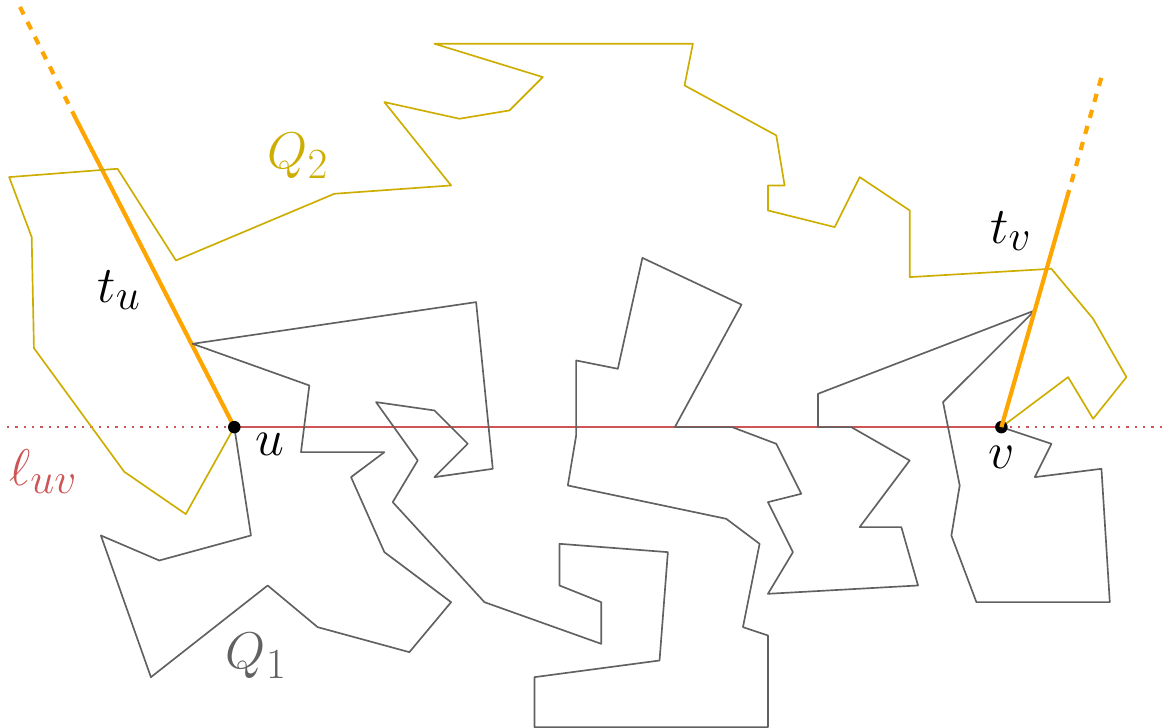}

    \caption{An example in which $t_u$ and $t_v$ do not intersect.}
    \label{fig:t_u-t_v no intersection} 
\end{figure}
If $t_u$ and $t_v$ intersect (in $\mathcal L^+$), we compute a \emph{leaning representation} $\Gamma_e$ of $P_e$; this is a straight-line drawing of $P_e$ that overlaps with points and segments of $Q_1$, and that is arbitrarily close to the actual drawing of $P_e$ that appears in the solution of the instance $\langle G,C,\Gamma_C \rangle$ we compute, if any such a solution exists; see \cref{fig:leaning-representation}. We use the leaning representation of $P_e$ to split $\langle G,C,\Gamma_C \rangle$ into subproblems. Indeed, if $t_u$ and $t_v$ lie on $\ell_{uv}$, since $S$ only contains points and segments on $\overline{uv}$, we let $p_w$ be any internal point of $\overline{uv}$, otherwise we let $p_w$ be the unique intersection point of $t_u$ and $t_v$. Let $\Gamma_e$ be the union of the segments $\overline{up_w}$ and $\overline{p_wv}$. 

\begin{restatable}{lemma}{lemmaleaningnecessary} \label{le:leaning-necessary}
A proper left drawing of $P_e$ exists if and only if: (1) $Q_1(\mathcal L)$ contains no bad chain, $S$ is non-empty, and $t_u$ and $t_v$ intersect; and (2) the leaning representation $\Gamma_e$ of $P_e$ does not intersect $\Gamma_C$ other than overlapping with points and/or segments in chains in $S$.
\end{restatable}


\begin{restatable}{lemma}{lemmaleaningcompute} \label{le:leaning-compute}
We can test in $O(n)$ time whether a proper left drawing of $P_e$ exists and, in the positive case, we can construct a leaning representation of $P_e$ within the same time bound.    
\end{restatable}







Having determined whether a proper left drawing and a proper right drawing of $P_e$ exist, we proceed as follows. If neither of such drawings exists, we  report that the instance is negative. If a proper left drawing exists and a proper right drawing does not, we orient $e$ so that it traverses $(u,v)$ from right to left. If a proper right drawing exists and a proper left drawing does not, we orient $e$ so that it traverses $(u,v)$ from left to right. If both proper left and right drawings of $P_e$ exist, then $\overline{uv}$ lies inside $\Gamma_C$, by \cref{le:donut}, and $e$ is not oriented.  

\subparagraph{Recursion.} Suppose that, while orienting $T$, we did not report the instance to be negative, as otherwise the algorithm is already concluded. If we did not orient any edge (which includes the case in which $T$ has no edge, that is, $\mathcal P=\emptyset$), we conclude and report that the instance is positive; this is the base case of our recursive algorithm. Indeed, every path in $\mathcal P$ can be drawn inside $\Gamma_C$ as a straight-line segment between its end-vertices; no two paths intersect since the embedding of $G$ is planar. If some edges of $T$ have been oriented, we contract the non-oriented edges in $T$ and obtain an oriented tree $T'$.\rfs{Each vertex of $T'$ corresponds to a non-empty set of internal faces of $G$ whose closures together bound a region of the plane homeomorphic to a closed disk.} Since $T'$ is a tree, it contains a source $s$. Let $e$ be any edge incident to $s$ in $T'$. Note that $e$ is also an edge of $T$. Let  $(u,v)$ be the internal edge of $O$ dual to $e$, and oriented from $u$ to $v$. Assume that $e$ is oriented so that it traverses $(u,v)$ from right to left, as the other case is analogous. Hence, the path $P_e=(u,w,v)$ in $\mathcal P$ corresponding to $e$ needs to have a proper left drawing.

Rather than directly constructing a proper left drawing of $P_e$, we use the leaning representation $\Gamma_e$ of $P_e$ in order to construct the smaller instances on which we apply recursion. Then a proper left drawing of $P_e$ inside $\Gamma_C$ is inserted, arbitrarily close to $\Gamma_e$, together with the recursively constructed solutions of the smaller instances. 
We observe the following.

\begin{restatable}{observation}{observationleaningverticesorder} \label{obs:leaning-vertices-order}
The order $u=v_1,v_2,\dots,v_k = v$ in $\Gamma_e$ (from $u$ to $v$) of the vertices of~$Q_1$ that lie on $\Gamma_e$ is the same order as the one in which such vertices appear \mbox{in $Q_1$ (from $u$ to $v$).}
\end{restatable}

For each pair of consecutive vertices $v_i$ and $v_{i+1}$ in the above sequence (see \cref{obs:leaning-vertices-order}) such that $(v_i,v_{i+1})$ is not an edge of $G$ we define an instance $\langle G_i, C_i, \Gamma_{C_i}\rangle$ as follows. If the part of $\Gamma_e$ between $v_i$ and $v_{i+1}$ does not contain in its interior the point $p_w$, let $C_i$ be the cycle composed of a new edge $(v_i,v_{i+1})$ and of the part $Q_1(v_i,v_{i+1})$ of $Q_1$ between $v_i$ and $v_{i+1}$; let $\Gamma_{C_i}$ be the straight-line planar drawing of $C_i$ composed of the segment $\overline{v_i v_{i+1}}$ and of the restriction of $\Gamma_C$ to $Q_1(v_i,v_{i+1})$; let $\mathcal P_i$ be the subset of $\cal P$ comprising the paths whose endpoints both belong to $C_i$, except the path connecting $v_i$ and $v_{i+1}$, if such a path belongs to $\mathcal P_i$; finally, let $G_i$ be the graph composed of $C_i$ and of the paths in $\mathcal P_i$.  
If the part of $\Gamma_e$ between $v_i$ and $v_{i+1}$ contains in its interior $p_w$, we place  $w$ at $p_w$; the cycle $C_i$ is composed of the new path  $(v_i,w,v_{i+1})$ and of $Q_1(v_i,v_{i+1})$; the drawing $\Gamma_{C_i}$ is composed of the segments $\overline{v_i p_w}$, $\overline{p_w v_{i+1}}$, and of the restriction of $\Gamma_C$ to $Q_1(v_i,v_{i+1})$; $\mathcal P_i$ is the subset of $\cal P$ comprising all the paths whose endpoints both belong to $C_i$; and $G_i$ is defined as above.
We call the constructed instances \emph{$Q_1$-instances}. We remark that the $Q_1$-instances contain edges that are not in $G$, namely each instance $\langle G_i,C_i,\Gamma_{C_i}\rangle$ contains a new edge $(v_i,v_{i+1})$, except, possibly, for the instance $\langle G_i,C_i,\Gamma_{C_i}\rangle$ such that the part of $\Gamma_e$ between $v_i$ and $v_{i+1}$ contains in its interior $p_w$. If it exists, such an instance contains the new edges $(v_i,w)$ and $(w,v_{i+1})$; also, in such an instance, $w$ is a neighbor of $v_i$ and $v_{i+1}$, while in $G$ it is a neighbor of $u$ and $v$.

We further create a \emph{$Q_2$-instance} $\langle G', C', \Gamma_{C'}\rangle$ as follows. Let $C'$ be the cycle composed of the path $P_e$ and of $Q_2$. Let $\Gamma_{C'}$ be the straight-line planar drawing of $C'$ composed of $\Gamma_e$ and of the restriction of $\Gamma_C$ to $Q_2$. 
Let $\cal P'$ be the subset of $\cal P$ comprising all the paths whose endpoints both belong to $C'$, except for $P_e$.
Finally, let $G'$ be the graph composed of $C'$ and of the paths in $\cal P'$. 
We have that the sets $\mathcal P_i$ and $\mathcal P'$ ``almost partition'' $\mathcal P$.

\begin{restatable}{lemma}{lemmapathpartition} \label{le:path-partition}
Suppose that a path $P\neq P_e$ in $\mathcal P$ does not connect two vertices $v_i$ and $v_{i+1}$ or that, if it does, the edge $(v_i,v_{i+1})$ belongs to $G$. Then $P$ belongs to a set $\mathcal P_i$ or to the set $\mathcal P'$.
\end{restatable}


The following lemma which allows us to apply recursion.

\begin{restatable}{lemma}{lemmapathoptimalwlog}\label{le:path-optimal-wlog}
We have that $\langle G,C,\Gamma_C\rangle$ is a positive instance of \PDEinF if and only if all the $Q_1$-instances and the $Q_2$-instance are positive instances of \PDEinF. Furthermore, all such instances can be defined in $O(n)$ time.
\end{restatable}

\subparagraph{Running time.} We conclude the proof of \cref{th:pde-in-f} by bounding the algorithm's runtime. For each path $P_e$ in $\mathcal P$ connecting two vertices $u$ and $v$ of $C$, we can test in $O(n)$ time whether~$P_e$ can be drawn inside $\Gamma_C$ as a segment $\overline{uv}$, by \cref{le:inside-detection}. In the negative case, by \cref{le:leaning-compute}, we can test in $O(n)$ time whether $P_e$ admits a proper left drawing or a proper right drawing. If all the tests fail, the algorithm terminates, otherwise we obtain an orientation for the edge $e$ of $T$ corresponding to $P_e$, or we have determined that such an edge is not oriented in $T$. Thus the orientation of $T$ can be found in $O(pn)$ time; note that $T$ has $O(p)$ vertices and edges. The tree $T'$ can be generated in $O(p)$ time and then a source $s$ in $T'$ can be found in $O(p)$ time. Constructing a leaning representation of a path $P_e$ corresponding to an edge $e$ incident to $s$ can be done in $O(n)$ time, by \cref{le:leaning-compute}. The instances on which recursion is applied can be constructed in $O(n)$ time, by \cref{le:path-optimal-wlog}. Thus, the process that leads to the recursion takes overall $O(pn)$ time. Since each instance on which recursion is applied has at most $n$ vertices and since the sum of the number of paths to be drawn in all the instances is at most $p-1$, the algorithm takes $O(p^2n)$ time. This concludes the proof of \cref{th:pde-in-f}.

\subparagraph{PDEoutF.} 
The proof of \cref{th:pde-out-f} is similar to the one of \cref{th:pde-in-f}. The main difference in \PDEoutF with respect to \PDEinF is the possible presence of paths in the form $P=(u,w,v)\in\mathcal{P}$ such that $P$ admits both a proper left and a proper right drawing but the straight-line segment $\overline{uv}$ does not lie outside $\Gamma_C$. Such paths are regarded as undirected in $T$.

\section{Bounded Vertex Cover Number}\label{se:vertex-cover}

In this section, we show that the \PDE problem is FPT parameterized by the vertex cover number $\vc(G)$ of the graph $G$, for biconnected $H$, both at fixed and at variable embedding. Recall that the \emph{vertex cover number} $\vc(G)$ of $G=(V,E)$ is the smallest size of a set $U\subseteq V$ such that every edge in $E$ is incident to at least one vertex in $U$.


\smallskip
We first show that the \PDEinF and \PDEoutF problems can be solved in FPT time parameterized by $\vc(G)$.  
%
%
For an instance $\langle G = (V,E),C,\Gamma_C \rangle$ of \PDEinF (or \PDEoutF), we use the following reduction rule to produce a kernel whose size is linear in $\vc(G)$. 


\begin{reductionrule} \label{rule:vc-deg2}
Let $\langle G=(V,E),C,\Gamma_C \rangle$ be an instance of \PDEinF (of \PDEoutF), and let $U$ be a vertex cover of $G$ of minimum size, i.e., $|U| = \vc(G)$.
For every pair of vertices $c,d\in U$, let $V_{cd}$ be the set of degree-$2$ vertices of $G$ in $V\setminus \{U \cup C\}$ with neighborhood~$\{c,d\}$. If~$|V_{cd}| \geq 2$, remove $|V_{cd}|-1$ vertices in $V_{cd}$ from $G$, as well as their incident edges.
\end{reductionrule}

Since $G$ is biconnected, it does not contain vertices of degree $1$. Also, the following lemma allows us to prove that, after applying \cref{rule:vc-deg2}, the number of vertices with degree at least $2$ left in the graph is a linear function of the vertex cover number.

\begin{lemma}[Lemma~11 of \cite{BekosLFGKLRT26,DBLP:conf/gd/BekosLF0KLRT24}] \label{le:degree-3-vc}
Let $X\subseteq V$ be a set of vertices in a planar graph $G=(V,E)$. The number of vertices in $V\setminus X$ that are connected to at least three vertices in $X$ is at most $2|X|$. Further, the number of pairs $(x,y)$ of vertices in $X$ such that $x$ and $y$ are the neighbors of a degree-$2$ vertex in $V\setminus X$ is at most $3|X|$. 
\end{lemma}

\begin{restatable}{lemma}{lekernelbound}\label{le:kernel}
Let $\langle G,C,\Gamma_C \rangle$ be an instance of \PDEinF (or \PDEoutF), where $G$ is an $n$-vertex graph with vertex cover number $k$. The instance $\langle G',C,\Gamma_C \rangle$ obtained by  applying \cref{rule:vc-deg2} is equivalent to $\langle G,C,\Gamma_C \rangle$, can be computed in $O(n)$ time, and has size $O(k)$, with $|V(C)| \leq 2k$.
\end{restatable}



Next, we show that \PDEinF and  \PDEoutF belong to\rfs{the class} $\exists\mathbb{R}$ both at fixed and\rfs{at} variable embedding.
We start by considering an instance $\langle G = (V,E),C,\Gamma_C \rangle$ of \PDEinF. 
Let $C = (v_1,v_2,\dots, v_{|C|})$, where\rfs{the vertex} $v_i$ immediately precedes\rfs{the vertex} $v_{i+1}$ in the clockwise order of the vertices of $C$ in $\Gamma_C $ (with $|C|+1 = 1$). Also, 
let $p_i = (x_{p_i}, y_{p_i})$ be the point where $v_i$ lies in~$\Gamma_C$. 
We assume that $\Gamma_C$ is oriented as prescribed by the planar embedding of $G$, if we are at fixed embedding, as otherwise $\langle G,C,\Gamma_C \rangle$ is a negative instance of \PDEinF. Clearly, this condition can be easily tested in $O(|C|)$ time.
For $\langle G, C, \Gamma_C  \rangle$ to be a positive instance of \PDEinF we only require that a drawing $\Gamma$ of $G$ exists satisfying the following constraints:

\crefname{enumi}{Constraint}{Constraints}
\Crefname{enumi}{Constraint}{Constraints}

\begin{enumerate}[(C1)]
\item \label{co:ofc} {\bf Outer face constraints:} each vertex $v_i \in C$ lies at point $p_i$;
\item \label{co:ivcc} {\bf In-polygon constraints:} each vertex $v \in V \setminus C$ lies strictly inside the polygon $\Gamma_C $; 
\item \label{co:evov} {\bf Non-overlap constraints:} for each edge $(u,v) \in E$, the closed straight-line segment $\overline{uv}$ does not include any point representing a vertex different from $u$ and $v$; and
\item \label{co:pc} {\bf Planarity constraints:} for each two edges $e_1,e_2 \in E$, the straight-line segments representing $e_1$ and $e_2$ do not intersect. 
\end{enumerate}

\noindent
Additionally, if $G$ has a fixed embedding, we require that $\Gamma$ satisfies the following constraints:

\begin{enumerate}[(C5)]
\item \label{co:emb} {\bf Embedding constraints:} for each vertex $v$, the cyclic order of the neighbors of $v$ in~$\Gamma$ is the same as the cyclic order of the neighbors of $v$ in the planar embedding of $G$.
\end{enumerate}

\removedforshortening{
\noindent 
The next two observations will be useful in modeling the existence of $\Gamma$ as a first-order formula in the Existential Theory of the Reals that satisfies Constraints C1--C5.
\begin{itemize}
	\item Since $G$ is connected, Constraint C3 prevents any two vertices from being assigned to the same point in $\Gamma$,  as this would imply an overlap between a vertex $v$ and the endpoint of the closed straight-line segment representing some edge of $G$ not incident to $v$. 
	\item Assuming Constraint C3 holds, Constraint C4 is satisfied if each pair of non-adjacent edges does not share an interior point; in fact, if two adjacent edges cross, then an edge would contain an endpoint of the other edge, violating Constraint C3, and if two non-adjacent edges cross at one of their endpoints then Constraint C3 would be violated. 
\end{itemize}

\noindent
Our main algorithmic ingredient is given by the following.
}

We only sketch here how to model Constraint C\ref{co:ivcc} as an ETR formula. For each vertex $v \in V$, we use the real variables $x_v$ and $y_v$ to represent the $x$- and $y$-coordinates of $v$, respectively, in the sought extension $\Gamma$ of $\Gamma_C$. Let $\mathcal{T}$ be any internal triangulation of the polygon $\Gamma_C$. Let $\Delta$ be a triangle of $\cal T$ bounded by vertices $p_i$, $p_j$, and $p_k$ of $\Gamma_C$. Let $\orient(p_i,p_j, p_k)$ be a polynomial expression representing the signed area of the triangle $\Delta$, which is positive (resp.\ negative) if $p_i$, $p_j$, and $p_k$ are encountered in this order when traversing the boundary of $\Delta$ counter-clockwise (resp.\ clockwise). A vertex $v$ lies in a point $p_v$ placed inside or on the boundary of $\Delta$ if and only 
if its orientations with respect to the three segments forming the boundary of $\Delta$ are all non-negative. Such a requirement can be modeled by the expression $\interior(\Delta,v):= \orient(p_i,p_j, p_v) \geq 0 \wedge \orient(p_j,p_k, p_v) \geq 0 \wedge \orient(p_k,p_i, p_v) \geq 0$. Hence, the constraint that a vertex $v$ should be placed in $\Gamma$ at a point in the interior or on the boundary of $\Gamma_C$ can tn be expressed by the formula  $inside_C(v): = \bigvee_{\Delta \in {\cal T}} \interior(\Delta,v)$, and Constraint C\ref{co:ivcc} can be modeled by the formula
$\bigwedge_{v \in V \setminus C} inside_C(v)$, given that Constraint C\ref{co:evov} guarantees that no vertex lies on the boundary of $\Gamma_C$.

\removedforshortening{Next, we show that \PDEoutF belongs to the class $\exists\mathbb{R}$.} 
Let now $\langle G = (V,E),C,\Gamma_C \rangle$ be an instance of \PDEoutF. 
\removedforshortening{Assume that $\Gamma_C$ is oriented as prescribed by the planar embedding of $G$, as otherwise $\langle G,C,\Gamma_C \rangle$ is a negative instance of \PDEoutF.} 
Note that, for $\langle G, C, \Gamma_C  \rangle$ to be a positive instance of \PDEoutF, we only require that there exists a drawing $\Gamma$ of $G$ satisfying Constraint C1 and Constraints C3--C5, and the constraint: 

\begin{enumerate}[(C6)]
\item \label{co:ovcc} {\bf Out-polygon constraints:} each vertex $v \in V \setminus C$ lies strictly outside the polygon~$\Gamma_C$.
\end{enumerate}

Such a constraint is easily modeled by the formula $\neg inside_C(v)$ and by requiring that no point representing a vertex in $V \setminus C$ belongs to $\Gamma_C$. The latter is however already achieved by Constraint C3. 
Altogether, we can show the following.

\begin{restatable}{lemma}{masejedamounaltronome}\label{le:PDEinAndoutF}
The \PDEinF (\PDEoutF, resp) problem belongs to the class $\exists\mathbb{R}$, both at fixed and at variable embedding. Moreover, for an instance $\langle G,C,\Gamma_C \rangle$ of \PDEinF (of \PDEoutF, resp.) an ETR formula modeling the existence of an extension $\Gamma$ of $\Gamma_C$ can be constructed using $O(|V(C)|)$ variables and~$O(|E(G)|^2)$ polynomials having maximum degree $4$ whose coefficients can be represented using~$O(1)$~bits.
\end{restatable}

\begin{restatable}{lemma}{fptvclemma}\label{th:pde-in-and-out}
\PDEinF and \PDEoutF are FPT with respect to the vertex cover number, both at fixed and at variable embedding.
\end{restatable}

\begin{proof}[Sketch]
Let $\langle G,C,\Gamma_C \rangle$ be an instance of \PDEinF (of \PDEoutF, resp.) with $\vc(G)=k$. 
First, we compute a minimum-size vertex cover of $G$ in $O(kn+1.274^{k})$ time~\cite{chenKanjXia:tcs10}.
Second, by \cref{le:kernel}, we construct a kernel for such an instance, whose size is linear with respect to $k$. Third, by \cref{le:PDEinAndoutF}, we model the existence of an extension for such a kernel as an ETR formula using $O(k)$ variables and $O(k^2)$ polynomials of maximum degree $4$ whose coefficients use $O(1)$ bits. Finally, we use \cref{thm:renegar} to test the satisfiability of such a formula.
\end{proof}


For instances with a prescribed embedding, \cref{{th:pde-in-and-out},le:fixed-emb-from-face-to-graph}, and the property that having a vertex cover of a certain maximum size $k$ is \mbox{monotone under subgraphs yield the following.}

\begin{theorem}\label{th:pde-fpt-vc}
{\sc PDE} is FPT with respect to the vertex cover number for instances $\langle G,H,\Gamma_H\rangle$ such that $H$ is biconnected and $G$ has a fixed planar embedding. 
\end{theorem}

In the rest of the section, we turn our attention to instances with a variable embedding and show that, also in this case, the problem is FPT w.r.t. to the vertex cover number.

\begin{theorem}\label{th:pde-fpt-vc-variable}
{\sc PDE} is FPT with respect to the vertex cover number for instances $\langle G,H,\Gamma_H\rangle$ such that $H$ is biconnected.
\end{theorem}

Let $G=(V,E)$ and let $U$ be a vertex cover of $G$ of minimum size $k$. At a high level, the proof is as follows. First, we obtain in linear time an instance $\langle G', H, \Gamma_H\rangle$ equivalent to $\langle G, H, \Gamma_H\rangle$ such that the number of vertices in $G' - H$ is in $O(k)$. Second, we identify $O(k)$ faces of $\Gamma_H$ in which the vertices of $G' - H$ can be assumed to lie in a straight-line planar drawing of $G'$ that extends~$\Gamma_H$, if any such a drawing exists. Third, we consider all possible placements of the vertices of $G'-H$ inside such faces, and use \cref{th:pde-in-and-out} to test whether a placement exists for which the corresponding instances of \PDEinF and \PDEoutF are all positive, in which case $\langle G, H, \Gamma_H\rangle$ is a positive instance of \PDE, or not, in which case $\langle G, H, \Gamma_H\rangle$ is a negative instance of \PDE. We now formalize this algorithm.

%
First, for every pair of vertices $u,v$ of $U$, we define $Z_{uv}$ as the set of degree-$2$ vertices in $V \setminus \{U \cup V(H)\}$  that have $u$ and $v$ as neighbors, and remove $|Z_{uv}|-1$ vertices in $Z_{uv}$ from~$G$. This procedure runs in $O(n)$ time and produces an equivalent instance $\langle G', H, \Gamma_H\rangle$. We verify in $O(n)$ time whether, for each vertex $w$ of $G'-H$, all the vertices in $H$ adjacent to $w$ appear on the boundary of a face of $\Gamma_H$. In the negative case, we reject the instance. 

\subparagraph{\bf Vertices of $G'$.} The vertex set of $G'$ contains $V(H)$ and at most $6k$ vertices not in $V(H)$, namely $k$ vertices in $U$,  at most $3k$ degree-$2$ vertices in $V \setminus \{U \cup V(H)\}$, and at most $2k$ vertices in $V \setminus \{U \cup V(H)\}$ of degree at least three, by \cref{le:degree-3-vc} with $X=U$.

\begin{figure}
    \centering
    \includegraphics[width=.7\linewidth]{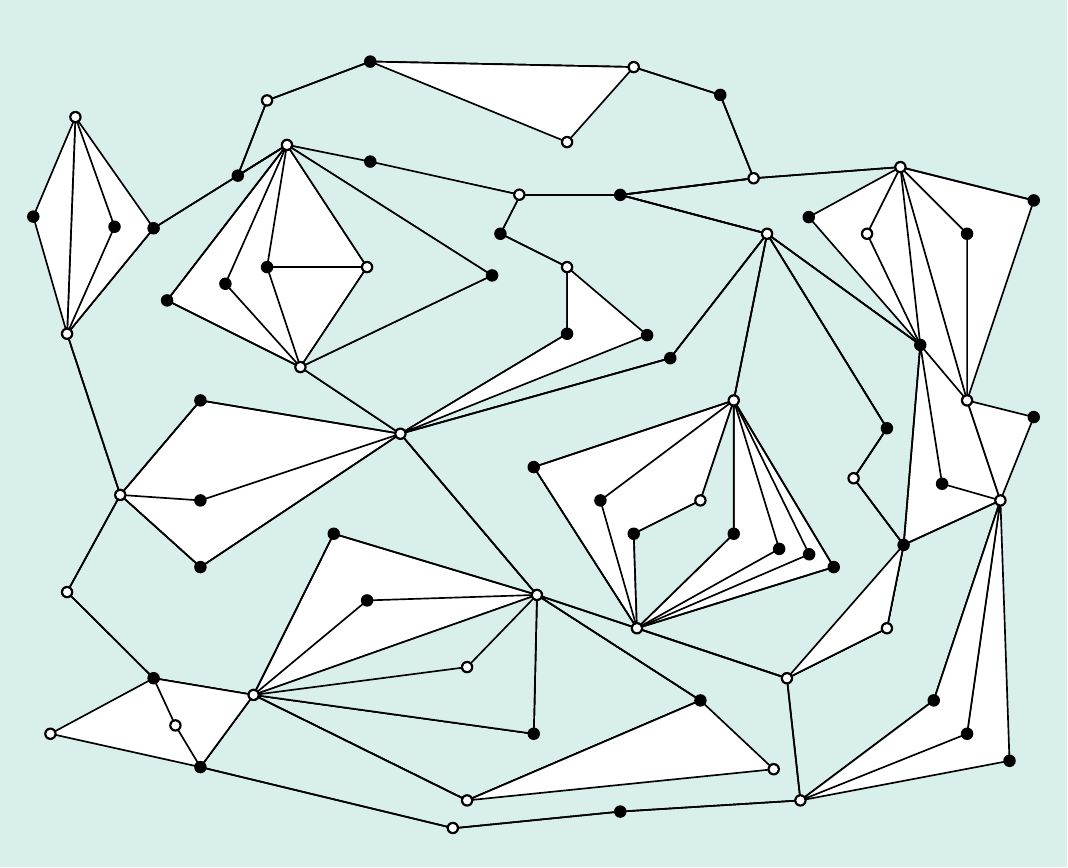}
    \caption{The drawing $\Gamma_H$, where the vertices of $H$ in $U$ are white, whereas those not in $U$ are black. Large faces are shaded blue, including the outer face. Small faces are white.}
    \label{fig:vc-faces}
\end{figure}


\subparagraph{\bf Vertices of $H$.} Let $V_{\geq 3}(H)$ be the set of vertices in $V(H) \setminus U$ adjacent to at least three vertices in $U$ and let $V_{=2}(H)$ be the set of vertices in $V(H) \setminus U$ adjacent to two vertices in $U$, which we call \emph{covered}. Since $U$ is a vertex cover of $G'$, every edge in $H$ has at least one endvertex in $U$. Since $H$ is biconnected, $\{V_{\geq 3}(H), V_{=2}(H)\}$ is a partition of $V(H) \setminus U$. By \cref{le:degree-3-vc} with $X = U$, we have $|V_{\geq 3}(H)|\leq 2k$, while $|V_{=2}(H)|$ might be in $\Omega(n)$, see \cref{fig:vc-faces}. 
%
The neighbors in $U$ of a covered vertex $w$ are in $H$; indeed, since $H$ is biconnected, $w$ has at least two neighbors in $H$ and since $w$ does not belong to $U$, all of its neighbors do. 



\subparagraph{Faces of $\Gamma_H$.} We aim to partition the faces of $\Gamma_H$ into two sets, see \cref{fig:vc-faces}. The first one contains $O(k)$ {\em interesting faces}, where we allow vertices of $G'-H$ to be inserted. The second one contains the {\em useless faces}, where we do not allow vertices of $G'-H$ to be inserted. We deem interesting every face of $\Gamma_H$ with at least $3$ vertices in $V(H)\cap U$ on its boundary; we call such faces \emph{large}. There are at most $2k$ large faces; this follows by \cref{le:degree-3-vc} with $X = U$ by inserting a star in each face of $\Gamma_H$. All other faces of $\Gamma_H$ have two vertices in $U$ on their boundaries, and thus they are delimited by $3$- or $4$-cycles; we call such faces \emph{small}. 
We call \emph{anchors} the covered vertices incident to small faces and adjacent to vertices in $G'-H$. We deem interesting the small faces that are incident to an anchor. We prove that the number of anchors is at most $11k$, hence the number of faces that we deem interesting because of them is at most $22k$; indeed, an anchor is a covered vertex, hence it has degree $2$ in $H$, thus it is incident to two faces of $\Gamma_H$. In order to prove that the number of anchors is at most $11k$, we consider three types of anchors. Those incident to large faces of $\Gamma_H$ with $4$ incident vertices are at most $2k$, as this is an upper bound on the number of large faces and $3$ of the $4$ vertices incident to each of such faces are in $U$. The number of anchors that are incident to large faces of $\Gamma_H$ with at least $5$ vertices and the number of anchors that are incident to small faces only are bounded by $7k$ and $2k$, respectively, by the following two lemmata. 

\begin{lemma}\label{le:number-anchors}
Let $G$ be a biconnected planar graph with vertex cover number $k$. The number of vertices incident to at least one face of length strictly larger than $4$ is at most $7k-12$.
\end{lemma}

\begin{proof}
Let $U$ be a vertex cover of $G$ of size $k$ and let $X$ be the set of vertices incident to at least one face of length strictly larger than $4$. Let $I = V(G)\setminus U$ be the subset of vertices of $G$ not in $U$, which is an independent set of $G$ given that $U$ is a vertex cover of $G$.
Since $ |X| \le |U| + |X\setminus U| = k + |X\setminus U|$, 
we only need to prove that $|X\setminus U| \leq 6k-12$.

Fix a planar drawing $\Gamma$ of $G$. For each vertex $x \in X\setminus U$, choose one face $f_x$ incident to $x$ whose length is strictly larger than $4$. On the boundary of $f_x$, the vertex $x$ has two neighbours; denote them by $u_x$ and $v_x$ and note that $u_x,v_x\in U$ and that $u_x \neq v_x$ given that the boundary of $f_x$ is a simple cycle. We now define an auxiliary planar multigraph $M$. Its vertex set is $U$, and for every vertex $x\in X\setminus U$, the graph $M$ contains an edge $e_x = (u_x,v_x)$. 
Hence, $|E(M)| = |X\setminus U|$. 
We construct a planar drawing $\Gamma_M$ of $M$ by placing its vertices as in $\Gamma$ and by drawing each edge $e_x$ in the plane along the two-edge path $(u_x, x, v_x)$. 

We claim that every edge $e_x$ of $M$ is incident to at least one face of $\Gamma_M$ of length at least $3$.
Indeed, the edge $e_x$ was defined using the face $f_x$ of $G$, whose length is strictly larger than $4$. Consider the side of $e_x$ corresponding to the face $f_x$. Suppose, for a contradiction, that this side becomes a digon face in $\Gamma_M$. Then this digon is bounded by $e_x$ and some parallel edge $e_y$, where $e_y = (u_x,v_x)$, for some vertex $y\in X\setminus U$. In $G$, the two corresponding paths $(u_x,x,v_x)$ and $(u_x,y,v_x)$ bound $f_x$, which hence has length $4$, a contradiction. Therefore, the side of $e_x$ corresponding to $f_x$ is not a digon. Consequently, $e_x$ is incident to a face of $\Gamma_M$ of length at least $3$. 

It remains to bound the number of edges incident to faces with at least $3$ vertices in a planar multigraph on $k$ vertices. For this, we use the following observation.

\begin{claim}
Let $M$ be a planar multigraph with $k$ vertices. Then the number of edges of $M$ that are incident to a face of length at least $3$ is at most $6k-12$.
\end{claim}

\begin{proof}
For any two vertices $u_x$ and $v_x$ of $M$, consider every maximal sequence of digon faces $f_1,\dots,f_d$ such that: (i) each face $f_i$ is incident to $u_x$ and $v_x$, and (ii) each two consecutive faces $f_i$ and $f_{i+1}$ share an edge. We can replace all the edges incident to the faces $f_1,\dots,f_d$ with a single edge $(u_x,v_x)$. Denote by $M'$ the resulting embedded graph. Note that, in $M'$, every face has at least three incident edges. The number of edges of $M$ that are incident to a face of length at least $3$ is at most twice the number of edges of $M'$. 

Let $\mu$ and $\phi$ be the number of edges and faces of $M'$, respectively.  
The sum of the number of edges incident to each face of $M'$ is equal to $2\mu$. Since each face of $M'$ has at least three incident edges, we have that $2\mu \geq 3\phi$. By Euler's formula, we thus get that $\phi = 2 - k + \mu$. Hence, $2 - k + \mu \leq \frac{2\mu}{3}$, which gives $\mu \leq 3k-6$.
\end{proof}

Applying the claim to our auxiliary graph $M$, and using the fact that every edge of $M$ is incident to a face of length at least $3$, we obtain $|X\setminus U| = |E(M)| \le 6k-12$. Therefore,
$|X| \leq |U| + |X\setminus U| \le k + 6k - 12 = 7k-12$.
\end{proof}

\begin{restatable}{observation}{anchorssmallobs}\label{obs:number-small}
The number of anchors incident to small faces only is at most $2k$.
\end{restatable}

Finally, for each pair $(u,v)$ of vertices in $V(H)\cap U$ that are incident to a common small face, if there exist small faces incident to $u$ and $v$ and not incident to anchors, we also deem interesting one of such faces and useless the others. Since there are at most $3k$ such pairs $(u,v)$, which can be proved again via \cref{le:degree-3-vc}, the number of small faces deemed interesting because of them is at most $3k$. This brings the total number of interesting faces to $27k$.


We now show that restricting the attention to interesting faces is not a loss of generality.

\begin{restatable}{lemma}{redrawlemma}\label{le:redraw}
If $\langle G', H, \Gamma_H\rangle$ is a positive instance of \PDE, there is a straight-line planar drawing $\Gamma'$ of $G'$ that extends $\Gamma_H$ where all the vertices in $G'-H$ \mbox{lie in interesting faces of $\Gamma_H$.}
\end{restatable}

\begin{proof}[Sketch]
In a solution, a useless face might only contain components of $G'-H$ attached to two vertices $u,v$ in $V(H)\cap U$. We redraw such components via Hong and Nagamochi's algorithm~\cite{DBLP:journals/dam/HongN08} in the interesting small face \mbox{with no anchor incident to $u$ and $v$.}
\end{proof}

To test whether $\langle G', H, \Gamma_H\rangle$ is a positive instance of \PDE we can then proceed as follows. We consider all possible $(27k)^{6k} \in k^{O(k)}$ assignments of the at most $6k$ vertices of $G'-H$ to the at most $27k$ interesting faces of $\Gamma_H$. For each such an assignment, we test, for each interesting face $f$, whether the instance of \PDEinF or of \PDEoutF (if $f$ is the outer face of $\Gamma_H$) determined by the assignment is a positive instance by \cref{th:pde-in-and-out}. If any such a test is successful, then we report that $\langle G', H, \Gamma_H\rangle$ and hence $\langle G, H, \Gamma_H\rangle$ is a positive instance of \PDE; otherwise, we reject the instance. This concludes the proof of \cref{th:pde-fpt-vc-variable}.

\section{Conclusions and Open Questions}\label{se:conclusions}

In this paper, we studied the \PDE problem, which asks to decide the extensibility of a partial straight-line planar drawing.
We proved that it is \NP-hard even for instances $\langle G,H,\Gamma_H\rangle$ in which $H$ is biconnected and only a set of length-$2$ paths remain to be drawn. We also tackled \PDE with $G$ having a fixed planar embedding and showed that the instances that we used to prove \NP-hardness are polynomial-time solvable in this setting. However, the computational complexity of \PDE for general instances $\langle G,H,\Gamma_H\rangle$ such that $H$ is biconnected and $G$ has a fixed planar embedding remains open. We find particularly fascinating the question of whether \PDE is polynomial-time solvable if $G$ is an internally-triangulated plane graph and $H$ consists solely of the cycle delimiting its outer face. Even proving membership in \NP~for this problem seems to be an elusive goal. 
Concerning parameterized algorithms, we proved that \PDE is FPT with respect to the vertex cover. It would be interesting to extend such a result to stronger parameterizations. In this direction, we proved that no $2^{o(tw)}$-time algorithm is possible, unless ETH fails, where $tw$ denotes the treewidth of the input graph $G$.

\bibliographystyle{abbrv}
\bibliography{bibliography}
\end{document}

%% file: fancyProblem.tex
\usepackage{framed}
\usepackage{tabularx}

\usepackage{tikz}
\usetikzlibrary{calc}

\newlength{\RoundedBoxWidth}
\newsavebox{\GrayRoundedBox}
\newenvironment{GrayBox}[1]%
   {\setlength{\RoundedBoxWidth}{0.95\columnwidth}
    \def\boxheading{#1}
    \begin{lrbox}{\GrayRoundedBox}
       \begin{minipage}{\RoundedBoxWidth}}%
   {   \end{minipage}
    \end{lrbox}
    \begin{center}
    \begin{tikzpicture}%
       \node(Text)[draw=black!20,fill=white,rounded corners,inner xsep=2ex,inner ysep=2ex,text width=\RoundedBoxWidth]
             {\usebox{\GrayRoundedBox}};
        \coordinate(x) at (current bounding box.north west);
        \node [draw=white,rectangle,inner sep=3pt,anchor=north west,fill=white]
        at ($(x)+(6pt,.75em)$) {\boxheading};
    \end{tikzpicture}
    \end{center}}

\newenvironment{defproblemx}[2]{\noindent\ignorespaces%
                                \FrameSep=6pt%
                                \parindent=6pt
                \vspace{-3mm}            
                \begin{GrayBox}{#1}%
                \begin{tabular*}{0.98\columnwidth}{!{\extracolsep{\fill}}@{} >{\itshape} p{#2} p{0.87\columnwidth} @{\hspace{.5em}}}%
            }{\\[-1.5ex]
                \end{tabular*}%
                \end{GrayBox}%
                \ignorespacesafterend
                \vspace{-4mm}
            }